\documentclass[superscriptaddress,aps,pra, twocolumn, showpacs,nofootinbib,longbibliography,notitlepage]{revtex4-2}
\usepackage{etex}
\usepackage{Physics}
\usepackage{amsmath,amssymb,amsthm,amstext,mathrsfs, mathtools}
\usepackage[colorlinks=true,citecolor=blue,urlcolor=blue]{hyperref}
\usepackage[pdftex]{graphicx}
\usepackage{times,txfonts}
\usepackage{braket}
\usepackage{color}
\usepackage{natbib}
\usepackage{amsmath,blkarray}
\usepackage{mathtools}
\usepackage{latexsym}
\usepackage{tabularx, booktabs}
\usepackage{graphics,epstopdf}
\usepackage{pgfplots}
\pgfplotsset{width=8 cm,compat=1.8}
\usepackage{graphicx}
\usepackage{float}
\usepackage{graphicx}
\usepackage{amsfonts}
\usepackage{subcaption}
\usepackage{enumerate}
\usepackage{color,soul}

\newcommand{\be}{\begin{equation}}
	\newcommand{\ee}{\end{equation}}
\newcommand{\ba}{\begin{eqnarray}}
	\newcommand{\ea}{\end{eqnarray}}

\newcommand{\half}{\frac{1}{2}}

\newtheorem{thm}{Theorem}

\begin{document}

\title{Unbounded Sharing of Nonlocality Using Projective Measurements}

\author{S. Sasmal}
\email{souradeep.007@gmail.com}
\affiliation{Indian Institute of Technology Hyderabad, Kandi, Sangareddy, Telengana 502285, India}

\author{S. Kanjilal}
\email{som.kanjilal@inl.int}
\affiliation{International Iberian Nanotechnology Laboratory, Braga, Portugal}
\affiliation{ International Institute of Information Technology, Hyderabad, Gachibowli, Telengana 500032, India}

\author{A. K. Pan }
\email{akp@phy.iith.ac.in}
\affiliation{Indian Institute of Technology Hyderabad, Kandi, Sangareddy, Telengana 502285, India}

\begin{abstract}
It is a common perception that a sharp projective measurement in one side of the Bell experiment destroys the entanglement of the shared state, thereby preventing the demonstration of sequential sharing of nonlocality. In contrast, we introduce a local randomness-assisted projective measurement protocol, enabling the sharing of nonlocality by an arbitrary number of sequential observers (Bobs) with a single spatially separated party Alice. Subsequently, a crucial feature of the interplay between the degrees of incompatibility of observables of both parties is revealed, enabling the unbounded sharing of nonlocality. Our findings, not only offer a new paradigm for understanding the fundamental nature of incompatibility in demonstrating quantum nonlocality but also pave a new path for various information processing tasks based on local randomness-assisted projective measurement. 
\end{abstract}

\maketitle

%%%%%%%%%%%%%%%%%%%%%%%%%%%%%%%%%%%%%%%%%%%%%%%%%%%%%%%%%%%%%%%%%%%%%%%%%%%%%%%%%%%%%%%%%%%%%%%%%%%%%%%%%%%%%%%%%%%%%%%%%%%%%%%%%%%%%%%%%%%%%%%%%%%%%%%%%%%%%%%%%%%%%%%%%%%%%%%%%%%%%%%%%%%%%%%%%%%%%%%%%%%%%%%%%%%%%%%%%%%%%%%%%%%%%%%%%%%%%%%%%%%%%%%%%%%%%%%%%%%%%%%%%%%%%%%%%%%%%%%%%%%%%%%%%%%%%%%%%%%%%%%%%%%%%%%%%%%%%%%%%%%%%%%%%%%%%%%%%%%%%%%%%%%%%%%%%%%%%%%%%%%%%%%%%%%%%%%%%%%%%%%%%%%%%%%%%%%%%%%%%%%%%%%%%%%%%%%%%%%%%%%%%%%%%%%%%%%%%%%%%%%%%%%%%%%%%%%%%%%%%%%%%%%%%%%%%%%%%%%%%%%%%%%%%%%%%%%%%%%%%%%%%%%%%%%%%%%%%%%%%%%%%%%%%%%%%%%%%%%%%%%%%%%%%%%%%%%%%%%%%%%%%%%%%%%%%%%%%%%%%%%%%%%%%%%%%%%%%%%%%%%%%%%%%%%%%%%%%%%%%%%%%%%%%%%%%%%%%%%%%%%%%%%%%%%%%%%%%%%%%%%%%%%%%%%%%%%%%%%%%%%%%%%%%%%%%%%%%%%%%%

\section{Introduction} 

Measurement plays a pivotal role in quantum theory which portrays a distinctive feature from its classical counterpart. In its standard formulation, quantum measurement is represented by a set of orthonormal projectors corresponding to a Hermitian operator. However, there exist more general measurements defined in terms of positive-operator-valued measures (POVMs) satisfying the completeness relation \cite{Busch1996book}. 

Since the projective measurement extracts more information from a quantum system compared to POVMs, one may surmise that projective measurement is more useful in information processing tasks. On the contrary, there exists quite a number of information processing tasks where POVMs showcase supremacy over projective measurements, such as, quantum state discrimination \cite{Bergou2013, Fields2020}, randomness certification \cite{Acin2016, Curchod2017, Andersson2018, Pan2021, Borkala2022}, quantum tomography \cite{Derka1998}, state estimation \cite{Bergou2010}, quantum cryptography \cite{Renes2004} and many more. Another pertinent example, relevant to this work, involves the sharing of quantum correlations among multiple sequential observers.  

Silva \emph{et al.} \cite{Silva2015} first demonstrated that two sequential observers on one side of the Bell experiment can share the nonlocality based on quantum violation of the Clauser-Horn-Shimony-Holt (CHSH) inequality \cite{Clauser1969}. Subsequently, the sharing of various other forms of quantum correlations have been explored \cite{Sasmal2018, Shenoy2019, Kumari2019, Mohan2019, Anwer2021}, given its potential applications in a wide-range of quantum information processing tasks. Consequently, it has been shown \cite{Brown2020} that an arbitrary number of independent observers can sequentially share nonlocality. Notably, all these schemes employ unsharp measurement for the sequential observers. 

In the conventional sharing protocol, there is one Alice (who always performs projective measurement) and an arbitrary $k$ number of independent sequential Bobs (say, Bob$^{k}$), who perform POVMs. An initial two-qubit entangled state is shared between Alice and first Bob (Bob$^{1}$). Upon receiving the respective subsystems, Alice and Bob$^{1}$ perform measurements. Bob$^{1}$ then replays the residual subsystem to the next sequential observer Bob$^{2}$ and the process continues until the quantum violation of Bell inequality is demonstrated between Alice and Bob$^k$. It is crucial to note that each sequential Bob must perform POVM measurements. This is because the projective measurement destroys the entanglement completely and hence  Bob$^{2}$ has no way to demonstrate the nonlocality. However, a recent demonstration \cite{Steffinlongo2022} illustrated that sharing of nonlocality can be achieved up to two sequential Bobs, even with projective measurements, provided Bob$^1$ employs suitable local randomness. Specifically, the sharing of nonlocality in \cite{Steffinlongo2022} is limited to two sequential Bobs. 

This study aims to demonstrate the sharing of nonlocality between Alice and arbitrary numbers of sequential Bobs through quantum violation of the CHSH inequality \cite{Clauser1969} using projective measurements. To accomplish this, for sequential Bobs, we introduce a distinctive form of qubit projective measurement aided with local randomness which we term as probabilistic projective measurement (PPM). The classical post-processing of projective measurements by Bob leads to a specific class of POVMs that plays a crucial role in exhibiting the unbounded sharing of nonlocality. Our work distinguishes itself from \cite{Steffinlongo2022} by unveiling the intricacies involved in the interplay between the incompatibility of Bob's observables and the commutativity of Alice's observable, coupled with the PPM scheme. 

It is worth noting that for a rank-1 projector $\pi_{i}=|\psi_{i}\rangle\langle\psi_{i}|$, there is a unique way to implement the projective measurement. Hence, for the projector $\pi_{i}$ acting on an initial quantum state $\rho$, the probability of obtaining the outcome $i$ is $p(i)=Tr[\pi_{i}\rho]$, and post-measurement state $|\psi_{i}\rangle$ is unique. In contrast, a rank-1 POVM element $\mathcal{E}_{i}$ can be implemented in various ways. Albeit the probability $p(i)=Tr[\mathcal{E}_{i}\rho]$ remains the same, irrespective of the way one implements the POVM element, crucially, the post-measurement state, in general, differs. Therefore, the degree of coherence or entanglement of the reduced state may vary depending on the way the measurement of the POVM element is implemented. This feature is key in demonstrating the sharing of quantum correlations by multiple sequential observers as the statistics of these sequential observers heavily depend on the procedures that are taken to implement the POVM element. The less the state is disturbed, the more sequential observers can share the quantum correlation. The PPM introduced in this paper achieves this by controlling the disturbance caused by the measurement.  The efficacy of PPM lies in its ability to sequentially share nonlocal quantum correlation among an arbitrary number of independent sequential observers. As indicated earlier, the interplay between the joint measurability of Bob's observables and the commutativity of Alice's observable plays a crucial role.

%%%%%%%%%%%%%%%%%%%%%%%%%%%%%%%%%%%%%%%%%%%%%%%%%%%%%%%%%%%%%%%%%%%%%%%%%%%%%%%%%%%%%%%%%%%%%%%%%%%%%%%%%%%%%%%%%%%%%%%%%%%%%%%%%%%%%%%%%%%%%%%%%%%%%%%%%%%%%%%%%%%%%%%%%%%%%%%%%%%%%%%%%%%%%%%%%%%%%%%%%%%%%%%%%%%%%%%%%%%%%%%%%%%%%%%%%%%%%%%%%%%%%%%%%%%%%%%%%%%%%%%%%%%%%%%%%%%%%%%%%%%%%%%%%%%%%%%%%%%%%%%%%%%%%%%%%%%%%%%%%%%%%%%%%%%%%%%%%%%%%%%%%%%%%%%%%%%%%%%%%%%%%%%%%%%%%%%%%%%%%%%%%%%%%%%%%%%%%%%%%%%%%%%%%%%%%%%%%%%%%%%%%%%%%%%%%%%%%%%%%%%%%%%%%%%%%%%%%%%%%%%%%%%%%%%%%%%%%%%%%%%%%%%%%%%%%%%%%%%%%%%%%%%%%%%%%%%%%%%%%%%%%%%%%%%%%%%%%%%%%%%%%%%%%%%%%%%%%%%%%%%%%%%%%%%%%%%%%%%%%%%%%%%%%%%%%%%%%%%%%%%%%%%%%%%%%%%%%%%%%%%%%%%%%%%%%%%%%%%%%%%%%%%%%%%%%%%%%%%%%%%%%%%%%%%%%%%%%%%%%%%%%%%%%%%%%%%%%%%%%%

 \section{Joint measurability, commutativity and CHSH violation} 
 
 While the connection between measurement incompatibility and CHSH violation is well-studied \cite{Wolf2009, Carmeli2012, Banik2013, Busch2013, Heinosaari2016}, here we show the connection between them is more nuanced than what has already been pointed out. The essential condition for CHSH violation requires that the shared state $\rho_{1}$ is entangled as well as the observables of both Alice and Bob must not be jointly measurable. 

Let's consider the scenario where Alice performs the projective measurements of two observables $A_x\equiv \qty{A_{\pm|x} \ \big| \ A_{\pm1|x}=\frac{1}{2}\qty(\mathbb{I}+A_x)}$ with $x\in \{0,1\}$ and  Bob performs unbiased POVMs $\qty{\mathcal{E}_{\pm|y} \ \big| \ \mathcal{E}_{\pm|y}=\frac{1}{2}(\mathbb{I}\pm \eta_{y} B_y)}$ with $y\in \{0,1\}$. Here $\eta_{y}$ is the unsharpness parameter of $B_{y}$.  It is important to note that incompatibility of Alice's sharp observables is ensured by their non-commutativity, i.e. $\braket{[A_0,A_1]}\neq 0$. in contrast, the incompatibility of Bob's unsharp observables is determined by $\eta_0^2+\eta_1^2>1$ \cite{Carmeli2012}. Now, lets delve into the analysis of how incompatibility plays a role for achieving the CHSH violation.

The quantum value of CHSH functional in the aforementioned scenario is given by 
\begin{equation} \label{chshunsh}
    \mathcal{I}=\Tr[\bigg\{\eta_{0} \ (A_0+A_1)\otimes B_0 + \eta_{1} \ (A_0-A_1) \otimes B_1\bigg\}\rho_1]
\end{equation}
We evaluate $\mathcal{I}$ by invoking an elegant sum-of-squares approach, as discussed in Appx.~\ref{ap1}, devoid of assuming the dimension of the quantum system. The condition for CHSH violation is derived as
\begin{eqnarray}\label{crit}
   \eta_{0} \sqrt{2+\langle \{A_0,A_1\}\rangle_{\rho_{1}}}+\eta_{1}\sqrt{2-\langle \{A_0,A_1\}\rangle_{\rho_{1}}} > 2
\end{eqnarray}
It is worth noting that the maximum CHSH violation ($2\sqrt{2}$) is achieved  when $\eta_{0}=\eta_{1}=1$ and Alice's and Bob's measurements are mutually anti-commuting, i.e., $\{A_0,A_1\}=\{B_0,B_1\}=0$, signifying maximal incompatibility. When $\eta_{0}=\eta_{1}$, we obtain the well-known critical value \cite{Andersson2005} (above which CHSH violation occurs) of the unsharpness parameter $\eta_c=\frac{1}{\sqrt{2}}$ which matches with the joint measurability condition between the unbiased POVMs \cite{Busch1986, Carmeli2012} corresponding to $B_{y}$. Importantly, attainment of such critical value requires Alice's observables to be mutually anti-commuting, i.e., $\{A_0,A_1\}=0$, else $\eta_c$ will be higher as illustrated in Fig.~\ref{figdinc}. The more the Alice's observables approaches towards commuting, the required value of Bob's unsharpness parameter $\eta$ becomes higher. Hence, there exists a trade-off between the degree of incompatibility of Bob's measurement and the anti-commutativity of Alice's observables (see Fig.~\ref{figdinc}). 

Now comes an important observation relevant to the present work.  The preceding argument takes a non-trivial twists when $\eta_{0}=1$ and $\eta_{1}=\eta$ i.e., one of Bob's measurement is sharp and the other remains unsharp. In that case, to obtain $\mathcal{I}>2$ in Eq. (\ref{crit}), one needs 
\begin{equation}\label{mucrit}
    \eta > \frac{2-\sqrt{2+\langle \{A_0,A_1\}\rangle_{\rho}}}{\sqrt{2-\langle \{A_0,A_1\}\rangle_{\rho}}}
\end{equation}

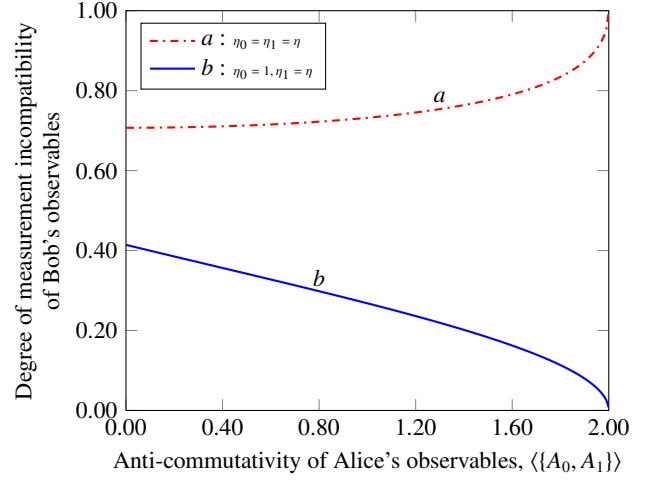
\begin{figure}[H]
	\centering
	\begin{tikzpicture}
	\begin{axis}[legend pos=north west, legend cell align=left, enlargelimits=false, xlabel={Anti-commutativity of Alice's observables, $\langle\{A_0,A_1\}\rangle$}, ylabel style ={align=center}, ylabel= {Degree of measurement incompatibility \\ of Bob's observables}, xticklabel style={ 
		/pgf/number format/fixed, /pgf/number format/fixed zerofill,
		/pgf/number format/precision=2
	}, scaled ticks=false, xtick={0,0.4,0.8,1.2,1.6,2.0}, yticklabel style={ 
		/pgf/number format/fixed, /pgf/number format/fixed zerofill,
		/pgf/number format/precision=2
	}, scaled ticks=false, ytick={0.00,0.2,0.4,0.6,0.8,1.0},xmin=0,xmax=2,ymin=0,ymax=1 
	]	

\addplot [domain=0:2, samples=200, red, thick, dashdotted] {(2)/(sqrt(2+x)+sqrt(2-x))};
\addlegendentry{\textit{a} : \tiny{$\eta_0=\eta_1=\eta$}}

\addplot [domain=0:2, samples=200, blue, thick] {(2-sqrt(2+x))/(sqrt(2-x))};
\addlegendentry{\textit{b} : \tiny{$\eta_0=1, \eta_1=\eta$}}

\node[above] at (130,75) {$a$};
\node[above] at (80,29) {$b$};

	\end{axis}
 
\end{tikzpicture}

	\caption{The plot illustrates the trade-off between incompatibility of Bob's observables and the anti-commutativity of Alice's observables in achieving the CHSH violation. The red dashed curve `a' signifies this trade-off when Bob's two observables are unsharp with equal unsharpness parameter ($\eta_{0}=\eta_{1}$). The blue curve `b' represents this trade-off in the scenario where one of Bob's observable is sharp ($\eta_{0}=1$) and the other is unsharp. While in the former case, the critical value of the degree of incompatibility monotonically increases with the anti-commutativity of Alice's observables, in the latter, degree of incompatibility monotonically decreases with the anti-commutativity.} 
	\label{figdinc}
\end{figure}

It is evident from the above Eq.~(\ref{mucrit}) that in order to demonstrate CHSH violation, the critical value of $\eta$ reduces to $\eta_c >0$ which is again the same as obtained \cite{Carmeli2012} for the measurement incompatibility POVMs corresponding to $B_0$ (Sharp) and $B_1$ (unsharp). Surprisingly, such a critical value is achieved only when Alice's measurements are almost commuting, i.e., $\{A_0,A_1\}\approx 2$ or $[A_0,A_1]\approx0$. Thus, the CHSH violation will be obtained if Alice measurements are non-commuting and $\eta>0$. The trade-off between the critical value of Bob's unsharpness parameter and anti-commutativity of Alice's observables in achieving the CHSH violation are depicted in Fig.~\ref{figdinc}.

An intriguing observation emerges from the preceding discussion is that violation of the CHSH inequality is possible even with nearly-commuting measurement observables for Alice and a nearly-zero unsharp parameter, signifying almost compatible measurements for both Alice and Bob. Such choices of measurements then allow the system to be minimally disturbed and preserve maximum quantum correlations, a key feature for a subsequent violation between Alice and the next sequential Bob. This feature along with the PPM scheme, plays a decisive role in demonstrating the sequential CHSH violations by an arbitrary number of observers.

%%%%%%%%%%%%%%%%%%%%%%%%%%%%%%%%%%%%%%%%%%%%%%%%%%%%%%%%%%%%%%%%%%%%%%%%%%%%%%%%%%%%%%%%%%%%%%%%%%%%%%%%%%%%%%%%%%%%%%%%%%%%%%%%%%%%%%%%%%%%%%%%%%%%%%%%%%%%%%%%%%%%%%%%%%%%%%%%%%%%%%%%%%%%%%%%%%%%%%%%%%%%%%%%%%%%%%%%%%%%%%%%%%%%%%%%%%%%%%%%%%%%%%%%%%%%%%%%%%%%%%%%%%%%%%%%%%%%%%%%%%%%%%%%%%%%%%%%%%%%%%%%%%%%%%%%%%%%%%%%%%%%%%%%%%%%%%%%%%%%%%%%%%%%%%%%%%%%%%%%%%%%%%%%%%%%%%%%%%%%%%%%%%%%%%%%%%%%%%%%%%%%%%%%%%%%%%%%%%%%%%%%%%%%%%%%%%%%%%%%%%%%%%%%%%%%%%%%%%%%%%%%%%%%

\section{Probabilistic projective measurement (PPM) scheme} 

Let Bob$^{k}$ performs the projective measurements of the observable $B_y \equiv \Big\{B_{\pm|y} \ \big| \ B_{\pm|y}=\frac{1}{2}\qty(\mathbb{I}\pm B_y)\Big\}$, where $B_{\pm|y}$ are the respective projectors corresponding to outcomes $\pm 1$. Additionally, assume that Bob$^{k}$ possesses a biased classical coin that yields heads with probability $\alpha_k$ and tails with probability $(1-\alpha_k)$. Bob uses the coin solely when he receives the input $y=1$, i.e., while performing the measurement of the observable $B_{1}$. If the head occurs, he implements the projective measurement $B_{\pm|1}$, inducing maximum disturbance to the state. If the tail occurs, he ``does nothing" ( i.e., performs $\mathbb{I}$), leaving the state undisturbed. This procedure implements the PPM in our work. The above scenario can be successfully captured by the three Kraus operators, given by 
\begin{equation}\label{fkpkmt}  
    \mathcal{K}^{k}_{1}=\sqrt{\alpha_k} \ B_{+|1} \ ; \ \mathcal{K}^{k}_{2}=\sqrt{\alpha_k} \ B_{-|1} \ ; \ \mathcal{K}^{k}_{3}=\sqrt{(1-\alpha_k)} \ \mathbb{I} \ ,
\end{equation}

Bob$^{k}$ collects his measurement statistics by classical post-processing of projective measurements. Hence, he effectively measures the following POVMs   
\begin{equation} \label{ppmim}
    \mathcal{E}^k_{+|1}=(\mathcal{K}^{k}_1)^{\dagger}(\mathcal{K}^{k}_1)+(\mathcal{K}^{k}_3)^{\dagger}(\mathcal{K}^{k}_3) \ ; \ \mathcal{E}^k_{-|1}=(\mathcal{K}^{k}_2)^{\dagger}(\mathcal{K}^{k}_2)
\end{equation}
which in turn gives
\begin{eqnarray}\label{ppm}
    \mathcal{E}_{+|1}^{k} = \alpha_k \ B_{+|1} + (1-\alpha_k) \ \mathbb{I}; \ \ \     \mathcal{E}_{-|1}^{k} = \alpha_k \ B_{-|1}
\end{eqnarray}
given that $\alpha_k\in[0,1]$, it follows that $\mathcal{E}^{k}_{\pm|1}\geq 0$. Moreover, by construction $\mathcal{E}_{+|1}^{k}+\mathcal{E}_{-|1}^{k}=\mathbb{I}$. The above argument can also be interpreted in the following manner: among the many possible implementations of a POVM, the POVMs in Eq. (\ref{ppm}) are implemented through the specific Kraus operators defined in Eq. (\ref{fkpkmt}).  
The probability of a outcome $\pm$ of Bob$^{k}$'s measurement is calculated as $p(\pm|B_1,\rho) = \Tr[\rho \ \mathcal{E}_{\pm|1}^k]$. Crucially, $p(\pm|B_1,\rho)$ remains unchanged regardless of how a POVM element is realised, although the reduced state may vary. This feature plays a distinctive role in our work, as mentioned earlier.

It is essential to underscore that the outlined PPM scheme differs fundamentally from the unsharp measurement formalism. In the latter, the unsharpness parameter originates from the indistinguishability of the apparatus's quantum states. In contrast, the parameter $\alpha_k$ in the PPM scheme solely arises from classical local randomness having no quantum origin. While any non-extremal POVMs can be simulated \cite{Oszmaniec2017} through the classical post-processing of projective measurements, it might be assumed that unbiased POVMs realised through the projective measurements could also lead to unbounded sharing of nonlocality. However, in such cases, the projective measurement completely disturbs the state, leaving no residual entanglement and, consequently, no sequential CHSH violation. This feature is straightforward and easily verifiable. The significance of our PPM scheme lies in its distinctiveness, enabling the  sharing of nonlocality for unbounded sequential Bobs.

%%%%%%%%%%%%%%%%%%%%%%%%%%%%%%%%%%%%%%%%%%%%%%%%%%%%%%%%%%%%%%%%%%%%%%%%%%%%%%%%%%%%%%%%%%%%%%%%%%%%%%%%%%%%%%%%%%%%%%%%%%%%%%%%%%%%%%%%%%%%%%%%%%%%%%%%%%%%%%%%%%%%%%%%%%%%%%%%%%%%%%%%%%%%%%%%%%%%%%%%%%%%%%%%%%%%%%%%%%%%%%%%%%%%%%%%%%%%%%%%%%%%%%%%%%%%%%%%%%%%%%%%%%%%%%%%%%%%%%%%%%%%%%%%%%%%%%%%%%%%%%%%%%%%%%%%%%%%%%%%%%%%%%%%%%%%%%%%%%%%%%%%%%%%%%%%%%%%%%%%%%%%%%%%%%%%%%%%%%%%%%%%%%%%%%%%%%%%%%%%%%%%%%%%%%%%%%%%%%%%%%%%%%%%%%%%%%%%%%%%%%%%%%%%%%%%%%%%%%%%%%%%%%%%%%%%%%%%%%%%%%%%%%%%%%%%%%%%%%%%%%%%%%%%%%%%%%%%%%%%%%%%%%%%%%%%%%%%%%%%%%%%%%%%%%%%%%%%%%%%%%%%%%%%%%%%%%%%%%%%%%%%%%%%%%%%%%%%%%%%%%%%%%%%%%%%%%%%%%%%%%%%%%%%%%%%%%%%%%%%%%%%%%%%%%%%%%%%%%%%%%%%%%%%%%%%%%%%%%%%%%%%%%%%%%%%%%%%%%%%%%
 
\section{Sharing of nonlocality using PPM}

We showcase how Bob's implementation of PPM, coupled with Alice's near-commuting projective measurements, facilitates the sharing of nonlocality between Alice and an unbounded number of independent sequential Bobs. As mentioned earlier, Bob employs the PPM scheme (refer to Eqs.~(\ref{ppm}) and (\ref{fkpkmt})) only when measuring the observable $B_{1}$; for $B_0$, he performs the standard projective measurement. The effective observable corresponding to $B_{1}$ can be expressed as $\mathcal{B}_1^k=\mathcal{E}_{+|1}^{k}-\mathcal{E}_{-|1}^{k}=\alpha_k B_1 + (1-\alpha_k) \mathbb{I}$. 

Given the above scenario, the quantum value of the CHSH functional between Alice and Bob$^k$ is given by
\begin{eqnarray}\label{chshs11}
\mathcal{I}^{k}&=&\Tr[\bigg\{\qty(A_0+A_1) \otimes B_0\bigg\} \rho_{k}] + \alpha_k \Tr[\bigg\{(A_0-A_1) \otimes  B_1\bigg\} \rho_k] \nonumber \\
&&+ (1-\alpha_k)\Tr[(A_0-A_1) \otimes \mathbb{I} \ \rho_k]
\end{eqnarray}
where $\rho_k$ represents the reduced state after $(k-1)^{th}$ Bob's measurement, evaluated by the Kraus evolution, as (see Appx.~(\ref{ap3})):
\begin{eqnarray}
\rho_{k} &=&\frac{1}{4} \Bigg[ (3-\alpha_{k-1}) \ \rho_{k-1} + \qty( \mathbb{I} \otimes B_{0}) \  \rho_{k-1} \qty( \mathbb{I} \otimes B_{0})  \nonumber \\
&&+ \alpha_{k-1} \ \qty( \mathbb{I} \otimes B_{1}) \  \rho_{k-1} \qty( \mathbb{I} \otimes B_{1}) \Bigg]
\end{eqnarray}

Note that, the CHSH value between Alice and Bob$^{1}$ is given by $\Tr[\{\qty(A_0+A_1) B_0\} \rho_{1}] + \alpha_1 \Tr[\bigg\{(A_0-A_1) B_1\} \rho_1] + (1-\alpha_1)\Tr[(A_0-A_1) \rho_1]$. For Bell diagonal states, the last term does not contribute to the CHSH value. In such a case, the CHSH value $\mathcal{I}^{1}$ resembles the expression given by Eq.~(\ref{crit}) with $\eta_0=1$ and $\eta_1=\alpha_1$. Hence, the critical value of $\alpha_1>0$ is sufficient to ensure the CHSH violation provided $[A_0,A_{1}]\to 0$, as determined by Eq.~(\ref{mucrit}).

For our purpose, we fix the following two-qubit  observables for Alice and Bob, and the two-qubit entangled state shared between them. 
\begin{eqnarray}
     A_0&=&\sin \delta \ \sigma_z +\cos \delta \ \sigma_x \ ; \ \ A_1 = -\sin\delta \ \sigma_z +\cos\delta \ \sigma_x \ ; \nonumber \\
    B_0 &=& \sigma_x \ ; \ \hspace{2.3 cm} {B}_1 =  \sigma_z \ ; \nonumber \\
    && \hspace{0.8 cm} \ket{\psi}_{1}=\cos\theta  \ \ket{00}+\sin\theta \ \ket{11} \label{sostpmmt}
\end{eqnarray}
where $0\leq \theta \leq \frac{\pi}{4}$ and $0\leq \delta\leq \frac{\pi}{2}$. %Now, in order to evaluate $\mathcal{I}^k$ given by Eq.~(\ref{chshs11}), we initially calculate the terms $\Tr[\qty{\qty(A_0+A_1) \otimes B_0}  \ \rho_{k}]$, $\Tr[\qty{(A_0-A_1) \otimes B_1} \ \rho_{k}]$ and $\Tr[\qty{(A_0-A_1) \otimes \mathbb{I}_2}  \ \rho_{k}]$. Subsequently, considering the following observables for Alice and Bob$^k$ and initial shared pure two-qubit entangled state: 
We evaluate the quantum expression for the CHSH functional as
\begin{eqnarray} \label{chshk}
\mathcal{I}^k\qty(\delta,\theta)&=& 2 \ \Bigg[\cos \delta  \sin2\theta \prod\limits_{j=1}^{k-1} \qty(1-\frac{\alpha_j}{2}) + \sin\delta\cos2\theta \nonumber \\
&& + \frac{\alpha_k}{2^{k-1}} \ \sin\delta \qty{1-2^{k-1}\cos2\theta} \Bigg] 
\end{eqnarray}
The details of the derivation are quite lengthy which is deferred to Appx.~\ref{ap3}.

To achieve CHSH violation for unbounded sequential Bobs, it is imperative to perturb the system minimally in order to preserve quantum correlations as much as possible requiring the value of $\alpha_{k}$ to be as minimum as possible. In view of Eq. (\ref{mucrit}) and the subsequent discussion, it indicates the need for a smaller degree of non-commutativity of Alice's observables. This feature is also depicted in Fig.~\ref{figdinc}. For the observables in Eq.~(\ref{sostpmmt}), one finds $[A_0,A_1]=2 i \sin2\delta \ \sigma_y$. Thus, we demonstrate the unbounded sequential sharing of nonlocality (i.e., $\mathcal{I}^k\qty(\delta,\theta)>2 \ \forall k$) in the specific regime of $\delta\rightarrow 0$, i.e., when Alice's observables are nearly commuting.

%%%%%%%%%%%%%%%%%%%%%%%%%%%%%%%%%%%%%%%%%%%%%%%%%%%%%%%%%%%%%%%%%%%%%%%%%%%%%%%%%%%%%%%%%%%%%%%%%%%%%%%%%%%%%%%%%%%%%%%%%%%%%%%%%%%%%%%%%%%%%%%%%%%%%%%%%%%%%%%%%%%%%%%%%%%%%%%%%%%%%%%%%%%%%%%%%%%%%%%%%%%%%%%%%%%%%%%%%%%%%%%%%%%%%%%%%%%%%%%%%%%%%%%%%%%%%%%%%%%%%%%%%%%%%%%%%%%%%%%%%%%%%%%%%%%%%%%%%%%%%%%%%%%%%%%%%%%%%%%%%%%%%%%%%%%%%%%%%%%%%%%%%%%%%%%%%%%%%%%%%%%%%%%%%%%%%%%%%%%%%%%%%%%%%%%%%%%%%%%%%%%%%%%%%%%%%%%%%%%%%%%%%%%%%%%%%%%%%%%%%%%%%%%%%%%%%%%%%%%%%%%%%%%%

\subsection{Main results} 

There exists pure entangled two-qubit states and suitable measurements for which nonlocality can be simultaneously shared between Alice and an arbitrary number of independent sequential Bobs through the proposed PPM scheme.

Details of the proof are presented in Appx.~\ref{ap3}. Here, we present the outline of the proof. We start by analysing the condition leading to the CHSH violation between Alice and Bob$^1$, i.e., $\mathcal{I}^1\qty(\delta,\theta)>2$. This implies that
\begin{equation}
   \alpha_1 > \frac{1-\sin(2\theta+\delta)}{2\sin\delta\sin^2\theta}
\end{equation}
Note that the lower bound of $\alpha_1$ will be zero if we choose $\theta=\frac{\pi}{4}-\frac{\delta}{2}$. Thus, Alice-Bob$^1$ will always obtain CHSH violation for all $\delta \in [0,\frac{\pi}{2}]$ for such choice of $\theta$.

Subsequently, sequential CHSH violations between Alice-Bob$^2$, Alice-Bob$^3$,.. Alice-Bob$^k$ narrowed down the range of $\delta$, thereby restricting the fraction of pure entangled two-qubit state for which each $\alpha_k$ lies in between $0$ and $1$. For instance, the CHSH violation between Alice-Bob$^2$, i.e., $\mathcal{I}^{2}{\qty(\delta,\theta)}>2$ implies $\delta$ must lie within a specific range, given by $0<\delta<\frac{\pi}{6}$ (see Appx.~\ref{ap3}). This implies only a certain class of pure entangled states will give rise to $\mathcal{I}^{2}{\qty(\delta,\theta)}>2$.

Furthermore, inserting $\theta= \frac{\pi}{4}-\frac{\delta}{2}$ in the quantum state for chosen measurement settings in Eq.~(\ref{sostpmmt}), we have
\begin{equation} \label{kpm1}
\mathcal{I}^{k}{\qty(\delta,\theta)}>2 \iff \alpha_k >  \frac{2^{k-1}\cos^2\delta \ \qty(1- \prod\limits_{j=1}^{k-1} \qty(1-\frac{\alpha_j}{2}))}{ \sin\delta \ \qty(1-2^{k-1} \sin\delta)}
\end{equation}
for arbitrary $k\geq 2$. Therefore, if the lower bound given by Eq.~(\ref{kpm1}) lies between $0$ and $1$ for arbitrary $k$ then $\mathcal{I}^{k}{\qty(\delta,\theta)}>2$. To be precise, we need to show that, for arbitrary $k\geq 2$ there exists an $\alpha_k$ such that
\begin{eqnarray} \label{kpm2}
    0< \frac{2^{k-1}\cos^2\delta \ \qty(1- \prod\limits_{j=1}^{k-1} \qty(1-\frac{\alpha_j}{2}))}{ \sin\delta \ \qty(1-2^{k-1} \sin\delta)} < \alpha_k \leq 1 
\end{eqnarray}

Note that, the lower bound in Eq.~(\ref{kpm2}) lies between $0$ and $1$ if $0<\delta<\sin^{-1}\frac{1}{2^{k-1}}$, consequently the range of $\theta$ is given by $\qty(\frac{\pi}{4}-\frac{1}{2}\sin^{-1}\frac{1}{2^{k-1}})<\theta< \frac{\pi}{4}$. This implies that the pure non-maximally two-qubit entangled state with a specified range of $\theta$  provides  $\mathcal{I}^{k}{\qty(\delta,\theta)}>2$. For the measurement settings and strategy adopted here, the range of $\theta$ narrows down with an increasing value of $k$. Nonetheless, we prove the following theorem. 
\begin{thm}\label{thm1}
 For any arbitrary $k\geq 2$ there exist pure entangled two-qubit states, characterized by the concurrence, $\mathcal{C}=\sin2\theta < 2^{1-k}\sqrt{4^{k-1}-1}$, such that there exists a sequence 
$ \{s_{1},s_{2},\ldots,s_{k}\}$ such that 
\begin{eqnarray} \label{kpmmt}
    0<\frac{2^{k-1}\cos^2\delta \ \qty(1- \prod\limits_{j=1}^{k-1} \qty(1-\frac{\alpha_j}{2}))}{ \sin\delta \ \qty(1-2^{k-1} \sin\delta)}<s_l<1
\end{eqnarray}
    for all $2\leq l\leq k$ and $s_1>0$.
\end{thm}

\textit{Sketch of the proof:-}  The detailed proof is quite lengthy and therefore deferred to  Appx.~\ref{ap3}. Here we explain the main steps to construct the proof.  Given an arbitrary $k$, we assume a sequence $\{s_{l}: \ 1\leq l\leq k\}$ satisfying the second inequality in the Eq.~(\ref{kpmmt}) and show that $(a)$ any finite sub-sequence is monotonically increasing, i.e. $s_{k}>s_{k-1}>\ldots>s_{1}$ and $(b)$ each term individually reaches zero in the limit i.e. $\alpha_1\to 0^{+} \implies s_{l}\to 0^{+}$ for all $1\leq l\leq k$ (see Appx.~\ref{ap3}). Since, the sequence is chosen in such a way that it satisfies the second inequality of Eq.~(\ref{kpmmt}), $(a)$ and $(b)$ together imply that it will lie between $0$ and $1$ for all $1\leq l\leq k$, given an arbitrary $k$, thus proving our claim.

%%%%%%%%%%%%%%%%%%%%%%%%%%%%%%%%%%%%%%%%%%%%%%%%%%%%%%%%%%%%%%%%%%%%%%%%%%%%%%%%%%%%%%%%%%%%%%%%%%%%%%%%%%%%%%%%%%%%%%%%%%%%%%%%%%%%%%%%%%%%%%%%%%%%%%%%%%%%%%%%%%%%%%%%%%%%%%%%%%%%%%%%%%%%%%%%%%%%%%%%%%%%%%%%%%%%%%%%%%%%%%%%%%%%%%%%%%%%%%%%%%%%%%%%%%%%%%%%%%%%%%%%%%%%%%%%%%%%%%%%%%%%%

\subsection{Comparison with related works}  

In \cite{Steffinlongo2022}, by invoking suitable local random variable $\lambda_{B}\in\{0,1\}$ (with the probability of occurring $\lambda_{B}=0$ is $q$) for Bob$^{1}$, a projective measurement based sequential sharing of nonlocality protocol has been devised. If $\lambda_{B}=0$, Bob$^{1}$ performs projective measurement of the observables $B_{0}^{\lambda_{B}=0}$ and $B_{1}^{\lambda_{B}=0}$, but If $\lambda_{B}=1$, Bob$^{1}$ performs $\mathbb{I}$ and the projective measurement of $B_{1}^{\lambda_{B}=1}$. The effective CHSH value between Alice and Bob$^1$ is then evaluated as a convex mixture of two separately computed CHSH values. There is another key step used in their protocol - an ad hoc manipulation of the post-measurement state. Specifically, a set of six unitary operations is applied to the post-measurement state pertaining to the outcome of Bob's measurements and the value of the random variable. The above steps enable them to construct the Kraus operators to derive the post-measurement state for Alice and Bob$^{2}$. This constitutes a different form of PPM scheme than ours. Note that, the sharing of nonlocality is restricted up to Bob$^{2}$ who performs a different set of observables than Bob$^{1}$. 

In contrast, in our PPM scheme, each sequential Bob performs projective measurements of the observable $B_{0}$ upon receiving the input $y=0$, but uses local randomness $\lambda_{B}\in \{0,1\}$ (with the probability of occurrence of $\lambda_{B}=0$ is $\alpha_{k}$) upon receiving the input $y=1$. Specifically, when $\lambda_{B}=0$ he applies $\mathbb{I}$; when $\lambda_{B}=1$ he performs the projective measurement of $B_{1}$. Effectively, for $y=1$, each Bob collects the statistics of the observable $\alpha_{k}B_{1} + (1-\alpha_{k})\mathbb{I}$. Note that there exist numerous sets of Kraus operators to realize the statistics of $\alpha_{k}B_{1} + (1-\alpha_{k})\mathbb{I}$ that remains the same irrespective of the way the Kraus operators are constructed. Crucially, the different sets of Kraus operators produce different post-measurement states having variable degrees of residual entanglement. In our PPM scheme, we construct an elegant set of Kraus operators that minimally disturb the system, ensuring sufficient residual entanglement to be shared by unbounded sequential Bobs.

An additional feature of our work is the construction of a two-parameter $(\alpha, v)$ POVM is that we have exhibited that if the same POVM is implemented through different Kraus operators (considering the implementation with two Kraus operators instead of three for the earlier projective measurement scheme), we obtain the previous result \cite{Brown2020} of unbounded sharing through unsharp measurement formalism as a special case. Specifically, we have derived a range of values of for the parameter of $v$ for which unbounded sharing is feasible. To be more precise, we have formulated the following theorem.

\begin{thm}\label{thm2}
 If the POVM defined by Eq.~(\ref{ppm}) is implemented through two Kraus operators given by  $\mathcal{K}_{+}^k= \sqrt{v (1-\alpha_k)+\alpha_k}B_{+|1}+\sqrt{v(1-\alpha_k)} B_{-|1}$ and $\mathcal{K}_{-}^k= \sqrt{(1-v)(1-\alpha_k)}B_{+|1}+\sqrt{1-v(1-\alpha_k)} B_{-|1}$ with the POVM elements $\mathcal{E}^k_{\pm}=(\mathcal{K}_{\pm}^k)^{\dagger}(\mathcal{K}_{\pm}^k)$, then for the case of $\theta=\frac{\pi}{4}$ and $0.058<v\leq 0.942$, an arbitrary number of Bobs can share the Bell nonlocality with a single Alice.
\end{thm}
\begin{proof}
    We have given the details of the proof in Appx.~\ref{ap4}. Here, we give the main arguments leading to the proof. In the case of $v\neq 0$ and $\theta=\frac{\pi}{4}$, we have shown that the Bell value $\mathcal{J}^k$ is lower bounded as follows,
\begin{equation}
\label{chshkx>0}
    \mathcal{J}^k \geq 2 \ \qty[\cos\delta\ \prod\limits_{j=1}^{k-1} \qty(1-\frac{\alpha_j^2 }{16v(1-v)}) +  \frac{\alpha_k \ \sin\delta}{2^{k-1}}]
\end{equation}
If the lower bound of $\mathcal{J}^k$ in Eq.~(\ref{chshkx>0}) is greater than two then 
\begin{equation}
\label{alphakx>0}
\alpha_{k} >\frac{2^{k-1}}{\sin\delta} \ \qty[1-\cos\delta \ \prod\limits_{j=1}^{k-1} \qty(1-\frac{\alpha_j^2 }{16v(1-v)})]
\end{equation}
Therefore, for arbitrary $k$, the sufficient condition to obtain CHSH violations of $k$ number of Bobs can now be translated into the question of existence of $\{\alpha_{l}:2\leq l\leq k\}$ such that
\small
\begin{equation}
\label{sufficientchshkx>0}
0 < \frac{2^{l-1}}{\sin\delta} \ \qty[1-\cos\delta \ \prod\limits_{j=1}^{l-1} \qty(1-\frac{\alpha_j^2 }{16v(1-v)})]<\alpha_{l}\leq \frac{1-4 v+4v^2}{1-3v+3v^2}
\end{equation}
\normalsize
and $\alpha_1>\tan\frac{\delta}{2}$ with $0 < \alpha_l \leq 1$. We then proceed similarly as Theorem $1$ to prove our claim.

\end{proof}

%%%%%%%%%%%%%%%%%%%%%%%%%%%%%%%%%%%%%%%%%%%%%%%%%%%%%%%%%%%%%%%%%%%%%%%%%%%%%%%%%%%%%%%%%%%%%%%%%%%%%%%%%%%%%%%%%%%%%%%%%%%%%%%%%%%%%%%%%%%%%%%%%%%%%%%%%%%%%%%%%%%%%%%%%%%%%%%%%%%%%%%%%%%%%%%%%%%%%%%%%%%%%%%%%%%%%%%%%%%%%%%%%%%%%%%%%%%%%%%%%%%%%%%%%%%%%%%%%%%%%%%%%%%%%%%%%%%%%%%%%%%%%%%%%%%%%%%%%%%%%%%%%%%%%%%%%%%%%%%%%%%%%%%%%%%%%%%%%%%%%%%%%%%%%%%%%%%%%%%%%%%%%%%%%%%%%%%%%%%%%%%%%%%%%%%%%%%%%%%%%%%%%%%%%%%%%%%%%%%%%%%%%%%%%%%%%%%%%%%%%%%%%%%%%%%%%%%%%%%%%%%%%%%%%%%%%%%%%%%%%%%%%%%%%%%%%%%%%%%%%%%%%%%%%%%%%%%%%%%%%%%%%%%%%%%%%%%%%%%%%%%%%%%%%%%%%%%%%%%%%%%%%%%%%%%%%%%%%%%%%%%%%%%%%%%%%%%%%%%%%%%%%%%%%%%%%%%%%%%%%%%%%%%%%%%%%%%%%%%%%%%%%%%%%%%%%%%%%%%%%%%%%%%%%%%%%%%%%%%%%%%%%%%%%%%%%%%%%%%%%%

\section{Summary and Outlook}

 In this work, by introducing an elegant PPM scheme, we have shown that it is possible to share nonlocality through the violation of CHSH inequality between Alice and an arbitrary number of independent sequential Bobs. Our approach thus circumvents the limitations encountered in the earlier work \cite{Steffinlongo2022}, where up to two sequential Bobs can share the nonlocality using local randomness assisted projective measurements. 

Our work distinguishes itself by revealing the underlying recipe for the possibility of unbounded sharing of nonlocality - the trade-off relationship between the degrees of incompatibility pertaining to the observables of both parties (as illustrated in Fig.~\ref{figdinc}). In order to demonstrate unbounded sharing of nonlocality, the shared state needs to be perturbed minimally. This necessitates selecting a measurement scheme that not only ensures the CHSH violation but also minimises the degree of incompatibility for both parties.

To further showcase the significance of our work, we bring out the way this particular feature of incompatibility plays a crucial role in the demonstration \cite{Brown2020} of unbounded sharing of nonlocality using unsharp measurement which overcomes the limitations encountered in earlier works of sharing \cite{Silva2015, Sasmal2018, Shenoy2019}. Subsequently, we have extended the proof of \cite{Brown2020} by proposing a two-parameter class POVMs for which unbounded sharing is possible. 

The significance of our protocol extends beyond sharing advantages; it enhances the foundational understanding of the intricacies linked with POVM implementation. This, in turn, heralds a new paradigm for information processing using POVMs. The versatility stems from the ability to implement a POVM in various ways, allowing for customization in the estimation of sequential statistics as per specific requirements. For example, our considered POVM $\qty{\mathcal{E}_{+}=\alpha B_{+|1} + v (1-\alpha) \mathbb{I}, \mathcal{E}_{-}=\alpha B_{-|1} + (1-v) (1-\alpha) \mathbb{I}}$ can be implemented using $(a)$ four Kraus operators (see Eq.~(\ref{fkpk}) of Appx.~\ref{ap2}) and $(b)$ two Kraus operators (see Eq.~(\ref{tpuk}) of Appx.~\ref{ap2}). While local randomness-assisted PPM can simulate case $(a)$, it is unable to simulate case $(b)$, thereby, unveiling the difference between the PPM scheme and the standard unsharp measurement formalism. This level of control over the implementation of the POVM elements paves the way for numerous applications. For instance, sharing scenario involving quantum communication advantage \cite{Miklin2020, Abhyoudai2023, Roy2023}, steering \cite{Sasmal2018}, contextuality \cite{Pan2019, Anwer2021}, sequential randomness certification \cite{Curchod2017, padovan2023}, which were explored under unsharp measurement formalism, can be revisited using the PPM scheme.

In essence, by unraveling the intricacies involved in the interplay between the incompatibility of Alice's and Bob's observables as well as by proposing the PPM scheme, we unlock a new paradigm for device-independent processing tasks using projective measurement with local randomness.

%%%%%%%%%%%%%%%%%%%%%%%%%%%%%%%%%%%%%%%%%%%%%%%%%%%%%%%%%%%%%%%%%%%%%%%%%%%%%%%%%%%%%%%%%%%%%%%%%%%%%%%%%%%%%%%%%%%%%%%%%%%%%%%%%%%%%%%%%%%%%%%%%%%%%%%%%%%%%%%%%%%%%%%%%%%%%%%%%%%%%%%%%%%%%%%%%%%%%%%%%%%%%%%%%%%%%%%%%%%%%%%%%%%%%%%%%%%%%%%%%%%%%%%%%%%%%%%%%%%%%%%%%%%%%%%%%%%%%%%%%%%%%%%%%%%%%%%%%%%%%%%%%%%%%%%%%%%%%%%%%%%%%%%%%%%%%%%%%%%%%%%%%%%%%%%%%%%%%%%%%%%%%%%%%%%%%%%%%%%%%%%%%%%%%%%%%%%%%%%%%%%%%%%%%%%%%%%%%%%%%%%%%%%%%%%%%%%%%%%%%%%%%%%%%%%%%%%%%%%%%%%%%%%%

\section{Acknowledgements} S.S. acknowledges the support from the project DST/ICPS/QuST/Theme 1/2019/4. SK acknowledges the hospitality of Indian Institute of Technology Hyderabad and the support from the FoQaCia project.  AKP acknowledges the support from the research grant MTR/2021/000908.

%%%%%%%%%%%%%%%%%%%%%%%%%%%%%%%%%%%%%%%%%%%%%%%%%%%%%%%%%%%%%%%%%%%%%%%%%%%%%%%%%%%%%%%%%%%%%%%%%%%%%%%%%%%%%%%%%%%%%%%%%%%%%%%%%%%%%%%%%%%%%%%%%%%%%%%%%%%%%%%%%%%%%%%%%%%%%%%%%%%%%%%%%%%%%%%%%%%%%%%%%%%%%%%%%%%%%%%%%%%%%%%%%%%%%%%%%%%%%%%%%%%%%%%%%%%%%%%%%%%%%%%%%%%%%%%%%%%%%%%%%%%%%%%%%%%%%%%%%%%%%%%%%%%%%%%%%%%%%%%%%%%%%%%%%%%%%%%%%%%%%%%%%%%%%%%%%%%%%%%%%%%%%%%%%%%%%%%%%%%%%%%%%%%%%%%%%%%%%%%%%%%%%%%%%%%%%%%%%%%%%%%%%%%%%%%%%%%%%%%%%%%%%%%%%%%%%%%%%%%%%%%%%%%%

\begin{widetext}
		\appendix

\section{Role of Measurement Incompatibility and Anti-Commuting measurements in CHSH violation:} \label{ap1}

Alice's measurements are described by projective measurements, given by $A_x\equiv \Big\{A_{\pm|x} \ \big| \ A_{\pm|x}=\frac{1}{2}\qty(\mathbb{I}\pm A_x)\Big\}$ with $x\in\{0,1\}$. Bob performs two smeared versions of projective measurement, given by $ \mathcal{B}_{y} =\eta_y B_y$ with $B_y \equiv \Big\{B_{\pm|y} \ \big| \ B_{\pm|y}=\frac{1}{2}\qty(\mathbb{I}\pm B_y)\Big\}$ and $y \in \{0,1\}$. In this scenario, the CHSH functional \cite{Clauser1969} is given as follows:
\begin{equation} \label{chshunsh}
    \mathscr{I}=\eta_0 \ (A_0+A_1)\otimes B_0 + \eta_1 \ (A_0-A_1) \otimes B_1
\end{equation}
The quantum value of the above CHSH functional is given by $\mathcal{I}=\Tr[\rho \ \mathscr{I}]$, where $\rho$ is the shared state between Alice and Bob. Now, in order to evaluate $\mathcal{I}$, we invoke the SOS method introduced in \cite{Pan2020}.

Let us define an operator $\gamma$ as follows:
\begin{equation} \label{sos1}
    \gamma = \sum\limits_{i=1}^{2}\frac{\omega_i}{2}(L_i)^{\dagger}(L_i)
\end{equation}
Now, without loss of generality, we can always choose $L_i$ and $\omega_i$ as follows
\begin{eqnarray} \label{sos2}
   && L_1 =\frac{\eta_0}{\omega_1}(A_0+A_1)\otimes \mathbb{I}-\mathbb{I}\otimes B_{0} \ \ ; \ \      L_2=\frac{\eta_1}{\omega_2}(A_0-A_1)\otimes \mathbb{I}-\mathbb{I}\otimes B_{1} \nonumber \\
 && \omega_1 = \eta_0 \ ||(A_0+ A_1)\rho^{\frac{1}{2}}||_F \ \ ; \ \ \omega_2 = \eta_1 \ ||(A_0- A_1)\rho^{\frac{1}{2}}||_F 
\end{eqnarray}
where $\norm{ \ \vdot \ }_F$ is the Frobenious norm and given by $||\mathcal{O}\rho^{\frac{1}{2}}||_F=\sqrt{\Tr[\mathcal{O}^{\dagger}\mathcal{O} \ \rho ]}$.

Note that by construction $L_i$'s are Hermitian operators, therefore, $\gamma$ is a positive semi-definite operator, i.e., $\Tr[\gamma \rho]\geq 0$. Now, if we put $L_i$ and $\omega_{i}$ into Eq.~(\ref{sos1}), we obtain the following
\begin{equation} \label{sos3}
    \Tr[\rho \mathscr{I}] = (\omega_1+\omega_2) - \Tr[\gamma \rho]
\end{equation}
Therefore, the quantum value of the CHSH functional is evaluated as follows: 
\begin{eqnarray}
    \mathcal{I} &=& \omega_1+\omega_2 - \Tr[\gamma \rho] \nonumber \\
    &=& \eta_0 \sqrt{2+\langle \{A_0,A_1\}\rangle_{\rho}}+\eta_1 \sqrt{2-\langle \{A_0,A_1\}\rangle_{\rho}} -\Tr[\gamma \rho] \nonumber \\
    &\leq& \eta_0 \sqrt{2+\langle \{A_0,A_1\}\rangle_{\rho}}+\eta_1 \sqrt{2-\langle \{A_0,A_1\}\rangle_{\rho}} \ \ \ \ [\text{since} \ \Tr[\gamma \rho]\geq 0]
\end{eqnarray}
The maximum quantum value of the CHSH functional will be then $ \mathcal{I}^{opt}=\eta_0 \sqrt{2+\langle \{A_0,A_1\}\rangle_{\rho}}+\eta_1 \sqrt{2-\langle \{A_0,A_1\}\rangle_{\rho}}$ when $Tr[\gamma \rho]=0$. Note that for $\eta_0=\eta_1=1$, the maximum quantum value is $2\sqrt{2}$ if Alice's and Bob's measurements are mutually anti-commuting, i.e., $\{A_0,A_1\}=\{B_0,B_1\}=0$.

%%%%%%%%%%%%%%%%%%%%%%%%%%%%%%%%%%%%%%%%%%%%%%%%%%%%%%%%%%%%%%%%%%%%%%%%%%%%%%%%%%%%%%%%%%%%%%%%%%%%%%%%%%%%%%%%%%%%%%%%%%%%%%%%%%%%%%%%%%%%%%%%%%%%%%%%%%%%%%%%%%%%%%%%%%%%%%%%%%%%%%%%%%%%%%%%%%%%%%%%%%%%%%%%%%%%%%%%%%%%%%%%%%%%%%%%%%%%%%%%%%%%%%%%%%%%%%%%%%%%%%%%%%%%%%%%%%%%%%%%%%%%%%%%%%%%%%%%%%%%%%%%%%%%%%%%%%%%%%%%%%%%%%%%%%%%%%%%%%%%%%%%%%%%%%%%%%%%%%%%%%%%%%%%%%%%%%%%%%%%%%%%%%%%%%%%%%%%%%%%%%%%%%%%%%%%%%%%%%%%%%%%%%%%%%%%%%%%%%%%%%%%%%%%%%%%%%%%%%%%%%%%%%%%

\section{Sequential Sharing Scenario: Single Alice - Multiple independent Bobs} \label{ap2}

Let the reduced state for $k^{th}$ Bob is $\rho_k$. In sequential measurement scenario, each Bob (say $(k-1)^{th}$ Bob) performs two different measurements on its own reduced state $\rho_{k-1}$. Therefore, the post-measurement states are given by
\begin{eqnarray}
   (\rho_{k})_{B^{k-1}_j} &\rightarrow{}& \sum_{b} \qty(\mathbb{I} \otimes \sqrt{B^{(k-1)}_{b|j}} ) \rho_{k-1} \qty(\mathbb{I} \otimes \sqrt{B^{(k-1)}_{b|j}}) \ \ \forall j\in\{0,1\}
\end{eqnarray}
Since all Bobs measure independently, the post measurement state $\rho_{k}$ after $(k-1)^{th}$ Bob's measurement is given by
\begin{equation} \label{indbob}
\rho_{k}= \frac{1}{2} \qty[ (\rho_{k})_{B^{k-1}_0} +(\rho_{k})_{B^{k-1}_1}]
\end{equation}

In the sequential scenario, Alice performs projective measurements and each Bob's first measurement $(B_0)$ is projective, while the second $(\mathcal{B}_1)$ is a two-parameter $(\alpha_k,v)$ POVM, formulated as follows:
\begin{eqnarray}
\mathcal{B}_1^k&=&\alpha_k \ B_1 + (2v-1)(1-\alpha_k) \ \mathbb{I} \\
      &\equiv& \Bigg\{ \ \ \mathcal{E}_{+|1}^k=\frac{1}{2} \ \bigg[\alpha_k \ B_1 + \big\{2v-\alpha_k \ (2v-1)\big\} \ \mathbb{I}\bigg],  \  \ \ \ \mathcal{E}_{-|1}^k = \frac{1}{2} \ \bigg[-\alpha_k \ B_1 +  \big\{1-(1-\alpha_k) \ (2v-1) \big\} \ \mathbb{I}\bigg] \ \ \Bigg\} \label{tpu}
\end{eqnarray}
with $0\leq v \leq 1$, $0\leq \alpha_k \leq 1$ and $B_{\pm|1}$ are the respective projectors corresponding to outcomes $\pm1$. $\mathcal{E}_{\pm}$ are two POVM-elements, called effect operator. It is crucial to note here that such POVM can be implemented through Kraus-operator formalism \cite{Busch1996book}. Now, the measurement of the observable $\mathcal{B}_1$ can be implemented in many different ways. Here, we will confine our studies for particular two cases:
\begin{enumerate} [(i)]
    \item POVM implemented by four Kraus operators:
    \begin{equation}\label{fkpk} 
    \mathcal{K}^{k}_{1}=\sqrt{\alpha_k} \ B_{+|1} \ ; \ \mathcal{K}^{k}_{2}=\sqrt{\alpha_k} \ B_{-|1} \ ; \ \mathcal{K}^{k}_{3}=\sqrt{v(1-\alpha_k)} \ \mathbb{I} \ ; \ \mathcal{K}^{k}_{4}=\sqrt{(1-v)(1-\alpha_k)} \ \mathbb{I} 
\end{equation}
Note that $v=1$ corresponds to the \textit{probabilistic projective measurement}. Here, $\mathcal{E}^k_{+}=(\mathcal{K}_1)^{\dagger}(\mathcal{K}_1)+(\mathcal{K}_3)^{\dagger}(\mathcal{K}_3)$ and $\mathcal{E}^k_{-}=(\mathcal{K}_2)^{\dagger}(\mathcal{K}_2)+(\mathcal{K}_4)^{\dagger}(\mathcal{K}_4)$.
\item POVM implemented by two Kraus operators:
\begin{equation}\label{tpuk}
    \mathcal{K}^{k}_{1}=\frac{m_1+m_2}{2} \ \mathbb{I}_2 +  \frac{m_1-m_2}{2} \ B_1 \ ; \ \ \mathcal{K}^{k}_{2}=\frac{n_1+n_2}{2} \ \mathbb{I}_2 +  \frac{n_1-n_2}{2} \ B_1 \ ;  
\end{equation}
where $m_1=\sqrt{x(1-\alpha_{k})+\alpha_{k}}$, $m_2=\sqrt{x(1-\alpha_{k})}$, $n_1=\sqrt{(1-\alpha_{k})(1-x)}$ and $n_2=\sqrt{1-x(1-\alpha_{k})}$. Note that $\alpha_k=\frac{1}{2}$ corresponds to the standard unsharp measurement \cite{Busch1986}.
\end{enumerate}

%%%%%%%%%%%%%%%%%%%%%%%%%%%%%%%%%%%%%%%%%%%%%%%%%%%%%%%%%%%%%%%%%%%%%%%%%%%%%%%%%%%%%%%%%%%%%%%%%%%%%%%%%%%%%%%%%%%%%%%%%%%%%%%%%%%%%%%%%%%%%%%%%%%%%%%%%%%%%%%%%%%%%%%%%%%%%%%%%%%%%%%%%%%%%%%%%%%%%%%%%%%%%%%%%%%%%%%%%%%%%%%%%%%%%%%%%%%%%%%%%%%%%%%%%%%%%%%%%%%%%%%%%%%%%%%%%%%%%%%%%%%%%%%%%%%%%%%%%%%%%%%%%%%%%%%%%%%%%%%%%%%%%%%%%%%%%%%%%%%%%%%%%%%%%%%%%%%%%%%%%%%%%%%%%%%%%%%%%%%%%%%%%%%%%%%%%%%%%%%%%%%%%%%%%%%%%%%%%%%%%%%%%%%%%%%%%%%%%%%%%%%%%%%%%%%%%%%%%%%%%%%%%%%%%%%%%%%%%%%%%%%%%

\section{Sharing of nonlocality using probabilistic projective measurement (POVM implemented by four Kraus operators)} \label{ap3}

The Kraus operators in this scenario is given by Eq.~(\ref{fkpk}). The average reduced state given by Eq.~(\ref{indbob}) is then expressed in terms of the Kraus operator as follows
\begin{eqnarray}\label{stategen2}
\rho_{k} &=&  \frac{1}{2} \qty[\sum_{b\in[+,-]} \qty(\mathbb{I} \otimes B^{k-1}_{b|0}) \ \rho_{k-1} \ \qty(\mathbb{I} \otimes B^{k-1}_{b|0})+\sum_{b\in\{1,2,3,4\}} \qty(\mathbb{I} \otimes \mathcal{K}^{k-1}_{b}) \ \rho_{k-1} \ \qty(\mathbb{I} \otimes \mathcal{K}^{k-1}_{b})] \nonumber \\
 &=& \frac{1}{2} \Bigg[ \Bigg\{\left( \mathbb{I} \otimes \frac{\mathbb{I} + B_{0}}{2} \right) \  \rho_{k-1} \  \left( \mathbb{I} \otimes \frac{\mathbb{I} + B_{0}}{2} \right) + \left( \mathbb{I} \otimes \frac{\mathbb{I} - B_{0}}{2} \right) \  \rho_{k-1} \  \left( \mathbb{I} \otimes \frac{\mathbb{I} - B_{0}}{2} \right) \Bigg\} \nonumber \\
 &&+\alpha_{k-1}\Bigg\{\qty( \mathbb{I} \otimes B_{0|1}) \  \rho_{k-1} \qty( \mathbb{I} \otimes B_{0|1}) + \qty( \mathbb{I} \otimes B_{1|1}) \  \rho_{k-1} \qty( \mathbb{I} \otimes B_{1|1}) \Bigg\}+ (1-\alpha_{k-1}) \ \rho_{k-1}\Bigg]  \nonumber \\
&=&\frac{1}{4} \Bigg[ (3-\alpha_{k-1}) \ \rho_{k-1} + \qty( \mathbb{I} \otimes B_{0}) \  \rho_{k-1} \qty( \mathbb{I} \otimes B_{0})   + \alpha_{k-1} \ \qty( \mathbb{I} \otimes B_{1}) \  \rho_{k-1} \qty( \mathbb{I} \otimes B_{1}) \Bigg]
\end{eqnarray}

The CHSH value between Alice-Bob$^k$ is given by\footnote{The CHSH functional is given by 
\begin{equation} \label{chshsf1}
    \mathscr{I}^{k} = \Big[A_0 + A_1\Big] \otimes B^{k}_0 + \Big[A_0 - A_1\Big] \otimes \mathcal{B}^{k}_1
    = \Big[A_0 + A_1\Big] \otimes B_0 + \Big[A_0 - A_1\Big] \otimes \Big[\alpha_k \ B_1 +(2v-1)(1-\alpha_k)\mathbb{I}\Big]
\end{equation}}
\begin{eqnarray} \label{chshs11}
    \mathcal{I}^{k} = \Tr[\mathscr{I}^{k} \ \rho_{k}] &=& \Tr[\bigg\{\qty(A_0+A_1) \otimes B_0\bigg\}  \ \rho_{k}] + \alpha_k  \ \Tr[\bigg\{(A_0-A_1) \otimes B_1\bigg\} \ \rho_{k}]  + (2v-1)(1-\alpha_k) \ \Tr[\bigg\{(A_0-A_1) \otimes \mathbb{I}_2\bigg\}  \ \rho_{k}]
\end{eqnarray}

Next, in order to evaluate $\mathcal{I}^k$ given by Eq.~(\ref{chshs11}), we evaluate the terms $\Tr[\qty{\qty(A_0+A_1) \otimes B_0}  \ \rho_{k}]$, $\Tr[\qty{(A_0-A_1) \otimes B_1} \ \rho_{k}]$ and $\Tr[\qty{(A_0-A_1) \otimes \mathbb{I}_2}  \ \rho_{k}]$ as follows:

\subsubsection*{Evaluation of \ $\Tr[\qty{\qty(A_0+A_1) \otimes B_0}  \ \rho_{k}]$}
\begin{eqnarray} \label{t11pm}
=&& \frac{3-\alpha_{k-1}}{4} \Tr[\qty{\qty(A_0+A_1) \otimes B_0} \ \rho_{k-1}] + \frac{1}{4} \Tr[\qty{\qty(A_0+A_1) \otimes B_0} \ \Big\{\qty(\mathbb{I} \otimes B_{0}) \ \rho_{k-1} \ \qty(\mathbb{I} \otimes B_{0})\Big\}] \nonumber \\
&&+\frac{\alpha_{k-1}}{4} \Tr[\qty{\qty(A_0+A_1) \otimes B_0}  \ \Big\{\qty(\mathbb{I} \otimes B_{1}) \ \rho_{k-1} \ \qty(\mathbb{I} \otimes B_{1})\Big\}] \nonumber \\
=&& \frac{2-\alpha_{k-1}}{2} \Tr[\qty{\qty(A_0+A_1) \otimes B_0} \ \rho_{k-1}]  \ , \ \ \ \text{Taking} \ \ \{B_0,B_1\}=0. \nonumber \\
=&&\Tr[\qty{\qty(A_0+A_1) \otimes B_0}  \ \rho_{1}] \ \prod\limits_{j=1}^{k-1} \qty(1-\frac{\alpha_j}{2})
\end{eqnarray}

\subsubsection*{Evaluation of $\Tr[\qty{(A_0-A_1) \otimes B_1} \ \rho_{k}]$}
\begin{eqnarray} \label{t121pm}
=&& \frac{3-\alpha_{k-1}}{4} \Tr[\qty{(A_0-A_1) \otimes B_1}  \ \rho_{k-1}] + \frac{1}{4} \Tr[\qty{(A_0-A_1) \otimes B_1} \ \Big\{\qty(\mathbb{I} \otimes B_{0}) \ \rho_{k-1} \ \qty(\mathbb{I} \otimes B_{0})\Big\}] \nonumber \\
&&+\frac{\alpha_{k-1}}{4} \Tr[\qty{(A_0-A_1) \otimes B_1} \ \Big\{\qty(\mathbb{I} \otimes B_{1}) \ \rho_{k-1} \ \qty(\mathbb{I} \otimes B_{1})\Big\}] \nonumber \\
    =&& \frac{1}{2^{k-1}} \ \Tr[\qty{(A_0-A_1) \otimes B_1} \ \rho_{1}]
    \end{eqnarray}

\subsubsection*{Evaluation of $\Tr[\qty{(A_0-A_1) \otimes \mathbb{I}_2}  \ \rho_{k}]$}
\begin{eqnarray} \label{t111pm}
=&& \frac{3-\alpha_{k-1}}{4} \Tr[\qty{(A_0-A_1) \otimes \mathbb{I}_2}   \ \rho_{k-1}] + \frac{1}{4} \Tr[\qty{(A_0-A_1) \otimes \mathbb{I}_2}   \ \Big\{\qty(\mathbb{I} \otimes B_{0}) \ \rho_{k-1} \ \qty(\mathbb{I} \otimes B_{0})\Big\}] \nonumber \\
&&+\frac{\alpha_{k-1}}{4} \Tr[\qty{(A_0-A_1) \otimes \mathbb{I}_2}   \ \Big\{\qty(\mathbb{I} \otimes B_{1}) \ \rho_{k-1} \ \qty(\mathbb{I} \otimes B_{1})\Big\}] \nonumber \\
=&& \Tr[\qty{(A_0-A_1) \otimes \mathbb{I}_2}  \ \rho_{k-1}] \nonumber \\
=&& \Tr[\qty{(A_0-A_1) \otimes \mathbb{I}_2}  \ \rho_{1}]
\end{eqnarray}

Then by invoking the above Eqs.~(\ref{t11pm}-\ref{t111pm}) into Eq.~\eqref{chshs11}, we obtain the following
\begin{equation} \label{chshki1}
\mathcal{I}^{k} = \Bigg\{\prod\limits_{j=1}^{k-1} \qty(1-\frac{\alpha_j}{2}) \Bigg\}  \ \Tr[\qty(A_0+A_1) \otimes B_0^k \ \ \rho_{1}]+ \frac{\alpha_k}{2^{k-1}} \ \Tr[\qty(A_0-A_1)\otimes B_1 \ \ \rho_{1}] + (1-\alpha_k)(2v-1) \ \Tr[\qty(A_0-A_1)\otimes \mathbb{I}_2 \ \ \rho_{1}]
\end{equation}

Now, let us consider the following observables and initial shared pure two-qubit entangled state
\begin{eqnarray}
   A_0&=&\sin \delta \ \sigma_z +\cos \delta \ \sigma_x \ ; \ \ A_1 = -\sin\delta \ \sigma_z +\cos\delta \ \sigma_x \ ; \nonumber \\
    B^k_0 &=& \sigma_x \ ; \ \ \mathcal{B}_1^k = (2v-1)(1-\alpha_k) \mathbb{I}+\alpha_k \ \sigma_z \ ; \label{obspm} \nonumber \\
    \ket{\psi}_{1}&=&\cos\theta  \ \ket{00}+\sin\theta \ \ket{11} \ ; \ \ 0\leq \theta \leq \frac{\pi}{4} \ ; \ \ 0\leq \delta\leq \frac{\pi}{2} \label{sostpm}
    \end{eqnarray}

Using Eq.~(\ref{sostpm}) we obtain:
\begin{equation} 
\mathcal{I}^k\qty(v,\delta,\theta)= 2 \ \Bigg[\cos \delta  \sin2\theta \prod\limits_{j=1}^{k-1} \qty(1-\frac{\alpha_j}{2}) + (2v-1)\sin\delta\cos2\theta + \frac{\alpha_k}{2^{k-1}} \ \sin\delta \qty{1-2^{k-1}(2v-1)\cos2\theta} \Bigg] 
\end{equation}

Before proving the Theorem $1$ given in the main text, let us look into the behaviours of $\mathcal{I}^{1}$ and $\mathcal{I}^{2}$ to investigate the allowed ranges of the variables.

The CHSH value between Alice-Bob$^1$ is
\begin{equation}
  \mathcal{I}^1 = 2 \ \Bigg[ \sin(2\theta-\delta) + 2v \sin\delta \cos 2\theta+ \alpha_1 \ \sin\delta\bigg\{1-(2v-1)\cos 2\theta\bigg\}\Bigg]  
\end{equation}
The CHSH violation for Alice-Bob$^1$, i.e., $ \mathcal{I}^1 > 2$ implies
\begin{equation}
   \alpha_1 > \frac{1-\sin(2\theta-\delta)-2v \sin\delta \cos 2\theta}{\sin\delta\bigg\{1-(2v-1)\cos 2\theta\bigg\}}
\end{equation}
Now, in order to minimise the lower bound of $\alpha_1$, we choose $\theta=\frac{\pi}{4}-\frac{\delta}{2}$ and $v=1$. In this case, the lower bound of $\alpha_1$ becomes, $\alpha_1>0$. In this case, Alice-Bob$^1$ will obtain CHSH violation for all $\delta \in [0,\frac{\pi}{2}]$. This implies we can take $\lim \alpha_1 \to 0^{+}$ and still obtain $\mathcal{I}^1>2$. Therefore we choose $\theta=\frac{\pi}{4}-\frac{\delta}{2}$ and $v=1$ as our preferred measurement settings. In this case the CHSH value between Alice-Bob$^k$ becomes:
\begin{equation} \label{chshkpmc31}
   \mathcal{I}^k \qty(v=1,\delta, \theta = \frac{\pi}{4}-\frac{\delta}{2}) =2 \ \qty[\cos^2\delta \ \prod\limits_{j=1}^{k-1} \qty(1-\frac{\alpha_j}{2}) +\sin^2\delta + \frac{\alpha_k}{2^{k-1}} \sin\delta \ \qty(1-2^{k-1} \sin\delta)]
\end{equation}
Now, the CHSH violation between Alice-Bob$^2$, i.e., $\mathcal{I}^2>2$ implies
\begin{eqnarray}
\alpha_2 &>& \frac{\alpha_1 \cos^2\delta}{ \sin\delta \ \qty(1-2 \sin\delta)} \ ; \ \text{the lower bound is between $0$ and $1$ for} \ 0<\delta<\frac{\pi}{6} \label{3povmchsh2pm}
\end{eqnarray}
Recall that, we need the lower bound in Eq~(\ref{3povmchsh2pm}) to be between $0$ and $1$ for plausibility of $\alpha_2$ yielding $\mathcal{I}^2>2$. In this case, the condition $0<\delta<\frac{\pi}{6}$ (in addition to $\alpha_1>0$) takes the lower bound in Eq~(\ref{3povmchsh2pm}) to be between $0$ and $1$. This implies only a certain class of pure entangled states will give rise to $\mathcal{I}_2>2$. Furthermore, $\min\limits_\delta\qty[ \frac{\cos^2\delta}{ \sin\delta \ \qty(1-2 \sin \delta)}]=2(\sqrt{3}+2)$ for $\delta=2 \tan^{-1}\qty(2+\sqrt{3}-\sqrt{6+4\sqrt{3}})\approx 0.2713$. Therefore, $\alpha_2 > 2(\sqrt{3}+2)\alpha_1$, which in turn fixes the upper bound of $0<\alpha_1\leq\frac{2-\sqrt{3}}{2}\approx 0.134$. Note that, it only changes puts an upper bound on $\alpha_1$ and the lower bound is still zero. increasing the number of Bobs puts stricter restrictions on the upper bounds of of $\delta$ and $\alpha_1$. Thus, in the limit of large $k$ both $\delta$ and $\alpha_1$ will be very close to zero. As stated in the main text, using Eq.~(\ref{chshkpmc31}), $\mathcal{I}^k > 2$ for arbitrary $k$ yields Eq.~(12) of the main text. 

Now, we will prove the Theorem $1$ stated in the main text. For the sake of clarity we again state the theorem below:

Theorem 1: For arbitrary $k$ there will always exist a fraction of pure entangled two-qubit states, characterized by the concurrence, $\mathcal{C}=\sin2\theta <2^{1-k}\sqrt{4^{k-1}-1}$, such that there exists a sequence $ \{s_{1},s_{2},\ldots,s_{k}\}$ such that
\begin{eqnarray} \label{kpm}
0<\frac{2^{k-1}\cos^2\delta \ \qty(1- \prod\limits_{j=1}^{k-1} \qty(1-\frac{\alpha_j}{2}))}{ \sin\delta \ \qty(1-2^{k-1} \sin\delta)}<s_l<1 \ \ \forall2\leq l\leq k, \ s_{1}>0
\end{eqnarray}

\begin{proof}

Without loss of generality, we take $\alpha_{1}=s_{1}>0$. For $l\geq 2$, consider the following sequence 
\begin{equation}
\label{94pm}
    s_l =\begin{cases}
        (1+\epsilon)\frac{2^{k-1}\cos^2\delta \ \qty(1- \prod\limits_{j=1}^{k-1} \qty(1-\frac{\alpha_j}{2}))}{ \sin\delta \ \qty(1-2^{k-1} \sin m)} \ \ & \text{where} \ s_{l-1}<1\\
        \infty \ \ & \text{otherwise}
    \end{cases}
\end{equation}
where $\epsilon>0$. Recall that, we need to prove $(a) \ \ s_{k}>s_{k-1}>\ldots>s_{1}$ and $(b)$ $\lim_{\alpha_{1}\to 0^{+}}s_{l}=0^{+}$. To prove $(a)$, first note that, $s_{2}=(1+\epsilon)\frac{s_{1}\cos^2\delta}{ \sin\delta \ \qty(1-2 \sin\delta)}>2(1+\epsilon)(\sqrt{3}+2)s_{1}$ since $\min\limits_\delta\qty[ \frac{\cos^2\delta}{ \sin\delta \ \qty(1-2 \sin\delta)}]=2(\sqrt{3}+2)$. For $l\geq 3$, we have the following

\begin{eqnarray}
    \frac{s_l}{s_{l-1}} &=& 2\times\frac{(1-2^{l-2}\sin\delta)}{(1-2^{l-1}\sin\delta)}\times\frac{\qty(1-\prod_{j=1}^{l-1}\qty(1-\frac{s_{j}}{2}))}{\qty(1-\prod_{j=1}^{l-2}\qty(1-\frac{s_{j}}{2}))} \label{96pm}
\end{eqnarray}

Note that $2^{l-2}\sin\delta < 2^{l-1}\sin\delta$, thus we have $\frac{(1-2^{l-2}\sin\delta)}{(1-2^{l-1}\sin\delta)}>1$. Furthermore, $\prod_{j=1}^{l-1}\qty(1-\frac{s_{j}}{2})<\prod_{j=1}^{l-2}\qty(1-\frac{s_{j}}{2})$ thus we have $\frac{\qty(1-\prod_{j=1}^{l-1}\qty(1-\frac{s_{j}}{2}))}{\qty(1-\prod_{j=1}^{l-2}\qty(1-\frac{s_{j}}{2}))}>1$. Therefore, we have $s_{l}>2s_{l-1}$ for $l\geq 3$. Thus, we have $s_{k}>s_{k-1}>\ldots>s_{1}$, proving $(a)$.

To prove $(b)$, first note that $s_{2}$ is a polynomial of $\alpha_{1}$ with lowest power of $\alpha_{1}$ being one. Thus, we have $\lim_{\alpha_{1}\to 0^{+}}s_{2} = 0^{+}$. Similarly, we have
\begin{eqnarray}
   s_{3}=\frac{4(1+\epsilon)\cos^{2}\delta}{\sin\delta(1-4\sin\delta)}\qty(\frac{s_{2}}{2}+\frac{s_{1}}{2}+\frac{s_{2}s_{1}}{4}), 
\end{eqnarray}
which is also a polynomial of $\alpha_{1}$ with lowest power of $\alpha_{1}$ being one. Therefore, $\lim_{\alpha_{1}\to 0^{+}}s_{3} = 0^{+}$. In particular, if for any $l$, if we assume each term in the sequence $\{s_{1},s_{2},s_{3},\ldots,s_{l-1}\}$ is a polynomial of $\alpha_{1}$ with lowest power of $\alpha_{1}$ being one then $\qty(1-\prod_{j=1}^{l-1}\qty(1-\frac{s_{j}}{2}))$ is also a polynomial of $\alpha_{1}$ with lowest power of $\alpha_{1}$ being one. To see this, let $\mathcal{P}_{n}(\alpha_{1})$ is the set of polynomials of $\alpha_{1}$ with lowest power of $\alpha_{1}$ being $n$. By assumption, for each $1\leq j\leq (l-1)$, we have $s_{j}\in \mathcal{P}_{1}(\alpha_{1})$. Furthermore, the constant coefficient (coefficient corresponding to zero-th power of $\alpha_{j}$)  of the polynomial $\prod_{j=1}^{l-1}\qty(1-\frac{s_{j}}{2})$ is one. Thus, $\qty[1-\prod_{j=1}^{l-1}\qty(1-\frac{s_{j}}{2})]\in \mathcal{P}_{1}(s_{1})$. Thus, $s_{l}$ given by Eq~(\ref{94pm}) is a polynomial of $\alpha_{1}$ where the lowest power of $\alpha_{1}$ is one. Consequently, we have $\lim_{\alpha_{1}\to 0^{+}}s_{l}=0^{+}$.
\end{proof}

%%%%%%%%%%%%%%%%%%%%%%%%%%%%%%%%%%%%%%%%%%%%%%%%%%%%%%%%%%%%%%%%%%%%%%%%%%%%%%%%%%%%%%%%%%%%%%%%%%%%%%%%%%%%%%%%%%%%%%%%%%%%%%%%%%%%%%%%%%%%%%%%%%%%%%%%%%%%%%%%%%%%%%%%%%%%%%%%%%%%%%%%%%%%%%%%%%%%%%%%%%%%%%%%%%%%%%%%%%%%%%%%%%%%%%%%%%
%%%%%%%%%%%%%%%%%%%%%%%%%%%%%%%%%%%%%%%%%%%%%%%%%%%%%%%%%%%%%%%%%%%%%%%%%%%%%%%%%%%%%%%%%%%%%%%%%%%%%%%%%%%%%%%%%%%%%%%%%%%%%%%%%%%%%%%%%%%%%%%%%%%%%%%%%%%%%%%%%%%%%%%%%%%%%%%%%%%%%%%%%%%%%%%%%%%%%%%%%%%%%%%%%%%%%%%%%%%%%%%%%%%%%%%%%%

\section{Sharing of nonlocality using two-parameter POVM implemented by two-Kraus operators~(proof of Theorem $2$):} \label{ap4}

The post measurement reduced state $\rho_{k}$ given by the Eq.~(\ref{indbob}) is evaluated by using the Kraus operator given by Eq.~(\ref{tpuk}) as follows

\begin{eqnarray}
\rho_{k} &=&  \frac{1}{2} \ \qty[ \sum_{b\in\{\pm\}} \qty(\mathbb{I} \otimes B^{k-1}_{b|0}) \ \rho_{k-1} \ \qty(\mathbb{I} \otimes B^{k-1}_{b|0})+ \sum_{b\in\{1,2\}} \qty(\mathbb{I} \otimes \mathcal{K}^{k-1}_{b}) \ \rho_{k-1} \ \qty(\mathbb{I} \otimes \mathcal{K}^{k-1}_{b})] \nonumber \\
   &=& \frac{1}{8}\Bigg[\qty{2+(m_1+m_2)^2+(n_1+n_2)^2} (\mathbb{I}\otimes\mathbb{I}) \ \rho_{k-1} \ (\mathbb{I}\otimes\mathbb{I}) + 2 \ (\mathbb{I}\otimes B_0) \ \rho_{k-1} \ (\mathbb{I}\otimes B_0) \nonumber \\
   &&+ \qty{m_1^2-m_2^2+n_1^2-n_2^2}\bigg\{(\mathbb{I}\otimes\mathbb{I}) \ \rho_{k-1} \ (\mathbb{I}\otimes B_1) +(\mathbb{I}\otimes B_1) \ \rho_{k-1} \ (\mathbb{I}\otimes\mathbb{I}) \bigg\} \nonumber \\
   &&+ \qty{(m_1-m_2)^2+(n_1-n_2)^2} (\mathbb{I}\otimes B_1) \ \rho_{k-1} \ (\mathbb{I}\otimes B_1)\Bigg] \nonumber \\ 
&=&\frac{2+\xi_{k-1}}{4} \ \rho_{k-1} + \frac{1}{4} \qty(\mathbb{I} \otimes B_{0}) \ \rho_{k-1} \ \qty(\mathbb{I} \otimes B_{0})  + \frac{1-\xi_{k-1}}{4} \qty(\mathbb{I} \otimes B_{1}) \ \rho_{k-1} \ \qty(\mathbb{I} \otimes B_{1})  \label{stategen2um}
\end{eqnarray}
The parameter $\xi_j$ is given as follows:
\begin{equation} \label{xijc3}
\xi_{j}  =
    \begin{cases}
      \sqrt{1-\alpha_j} & \ \text{if} \ v=0 \ \text{or} \ v=1 \\
      \sqrt{1-\alpha_j^2} & \ \text{if} \ v=\frac{1}{2} \\
       v\sqrt{1+q_1}+(1-v)\sqrt{1-q_2} & \ \text{if} \  0<v<\frac{1}{2} \ \text{and} \ \frac{1}{2}<v<1 
    \end{cases}       
\end{equation}
with $q_1=\alpha_j \ \qty(\frac{1-2v}{v})-\alpha_j^2 \ \qty(\frac{1-v}{v})$ and $q_2=\alpha_j \ \qty(\frac{1-2v}{1-v})+\alpha_j^2 \ \qty(\frac{v}{1-v})$.

Note that for $v=\frac{1}{2}$, $\xi_j=\sqrt{1-\alpha_j^2}$ implies one-parameter unbiased POVM, known as the unsharp measurement formalism with unsharp parameter $\alpha_j$.

The CHSH value corresponding to Alice and $k^{th}$ Bob is given by 
\begin{equation} \label{chshs1um}
    \mathcal{I}^{k} = \Tr[\bigg\{(A_0+A_1) \otimes B_0\bigg\}  \ \rho_{k}] + \alpha_k \ \Tr[\bigg\{(A_0-A_1) \otimes B_1 \bigg\}  \ \rho_{k}] + (2v-1)(1-\alpha_k) \ \Tr[\bigg\{(A_0-A_1) \otimes \mathbb{I}_2 \bigg\}  \ \rho_{k}] 
\end{equation}

Next, in order to evaluate $\mathcal{I}^k$ given by Eq.~(\ref{chshs1um}), we evaluate the terms $\Tr[\bigg\{(A_0+A_1) \otimes B_0\bigg\}  \ \rho_{k}]$, $\Tr[\bigg\{(A_0-A_1) \otimes B_1 \bigg\}  \ \rho_{k}]$, and $\Tr[\bigg\{(A_0-A_1) \otimes \mathbb{I}_2 \bigg\}  \ \rho_{k}]$.

\subsubsection*{Evaluation of \ $\Tr[\bigg\{(A_0+A_1) \otimes B_0\bigg\}  \ \rho_{k}]$}
\begin{eqnarray} \label{t11um}
=&& \frac{2+\xi_{k-1}}{4} \Tr[\bigg\{(A_0+A_1) \otimes B_0\bigg\}  \ \rho_{k-1}] + \frac{1}{4} \Tr[\bigg\{(A_0+A_1) \otimes B_0\bigg\}  \ \Big\{\qty(\mathbb{I} \otimes B_{0}) \ \rho_{k-1} \ \qty(\mathbb{I} \otimes B_{0})\Big\}] \nonumber \\
&&+\frac{1-\xi_{k-1}}{4} \Tr[\bigg\{(A_0+A_1) \otimes B_0\bigg\}  \ \Big\{\qty(\mathbb{I} \otimes B_{1}) \ \rho_{k-1} \ \qty(\mathbb{I} \otimes B_{1})\Big\}] \nonumber \\
=&& \frac{2+\xi_{k-1}}{4} \Tr[\bigg\{(A_0+A_1) \otimes B_0\bigg\}  \ \rho_{k-1}] + \frac{1}{4} \Tr[\bigg\{(A_0+A_1) \otimes B_0\bigg\}  \ \rho_{k-1}] \nonumber \\
&& + \frac{1-\xi_{k-1}}{4} \Tr[\bigg\{(A_0+A_1) \otimes B_1B_0B_1\bigg\}  \ \rho_{k-1}] \nonumber \\
=&& \frac{3+\xi_{k-1}}{4} \Tr[\bigg\{(A_0+A_1) \otimes B_0\bigg\}  \ \rho_{k-1}] + \frac{1-\xi_{k-1}}{4} \Tr[\Big\{(A_0+A_1) \otimes (B_1\{B_0,B_1\}- B_0)\Big\}  \ \rho_{k-1}] \nonumber \\
=&& \Bigg\{\frac{1}{2^{k-1}} \ \prod\limits_{j=1}^{k-1} (1+\xi_j) \Bigg\}  \ \Tr[\bigg\{(A_0+A_1) \otimes B_0\bigg\} \ \rho_{1}] \ , \ \ \ \text{Taking} \ \ \{B_0,B_1\}=0.
\end{eqnarray}

With a similar straight forward operator algebra, other terms are evaluated as given by 
\begin{eqnarray} \label{t13um}
\Tr[\bigg\{(A_0-A_1) \otimes B_1 \bigg\}  \ \rho_{k}] &=& \frac{1}{2^{k-1}} \ \Tr[\bigg\{(A_0-A_1) \otimes B_1 \bigg\}  \ \rho_{1}] \nonumber \\
\Tr[\bigg\{(A_0-A_1) \otimes \mathbb{I}_2 \bigg\}  \ \rho_{k}]&=&\Tr[\bigg\{(A_0-A_1) \otimes \mathbb{I}_2 \bigg\}  \ \rho_{1}]
\end{eqnarray}

For the observables and initial shared state given by Eq.~(\ref{obspm}), we obtain following:
\begin{equation} \label{chsho1um}
\mathcal{I}^k(v,\delta,\theta)= \frac{2 }{2^{k-1}}\Bigg[\cos \delta  \sin2\theta \prod\limits_{j=1}^{k-1} \qty(1+\xi_j) + 2^{k-1} (2v-1) \sin\delta\cos2\theta + \alpha_k \sin\delta\qty{1-2^{k-1}(2v-1)\cos 2\theta} \Bigg] 
\end{equation}

Now, considering the shared state to be maximally entangled two-qubit state and taking the range $0<v<\frac{1}{2}$, the CHSH value between Alice and Bob$^k$ is given by
\begin{eqnarray}
    \mathcal{J}^k=\mathcal{I}^k\qty(0<v<\frac{1}{2},\delta,\theta=\frac{\pi}{4})= \frac{2}{2^{k-1}}\Bigg[\cos \delta \prod\limits_{j=1}^{k-1} \qty(1+\xi_j) +\alpha_k \ \sin\delta \Bigg] \label{chshkc2um}
\end{eqnarray}
Here $\xi_j = v\sqrt{1+q_1}+(1-v)\sqrt{1-q_2}$ with  $q_1=\alpha_j \ \qty(\frac{1-2v}{v})-\alpha_j^2 \ \qty(\frac{1-v}{v})$ and $q_2=\alpha_j \ \qty(\frac{1-2v}{1-v})+\alpha_j^2 \ \qty(\frac{v}{1-v})$. Note that, $\xi_{j}$ is symmetric about half i.e. $\xi_{j}(v)=\xi_{j}(1-v)$. Thus, without the loss of generality, we can assume that $v$ lies between zero and half. Note that, for $v=\frac{1}{2}$ the scenario described here reduces to the case considered in \cite{Brown2020}. Here we deal with the situation where $0<v<1/2$. We shall specifically demonstrate that, given an arbitrary $k$, for a suitably chosen $x$ between zero and half, one can find a sequence $\{\alpha_{j}:1\leq j\leq k\}$ such that $\mathcal{I}^{j}>2$ for all $1\leq j\leq k$. 

For $0<v<\frac{1}{4} (2-\sqrt{2})$, if we choose $\alpha_j$ such that it satisfies $\Big\{\qty(0 \leq \alpha_j < 1-\frac{1+2\sqrt{1-8v(1-v)}}{2(1-v)}) \ \cup \  \qty(1-\frac{1-2\sqrt{1-8v(1-v)}}{2(1-v)}<\alpha_j\leq \frac{1-2v}{1-v})\Big\}$, and for $\frac{1}{4} (2-\sqrt{2})\leq x < \frac{1}{2}$, if we choose $0\leq \alpha_j \leq \frac{1-2v}{1-v}$, then such choices imply $0\leq q_1\leq 1$. Moreover, $0< x < \frac{1}{2} $ and $0\leq\alpha_j\leq 1$ implies $0\leq q_2 \leq 1$. Therefore, when $0\leq \alpha_{j}\leq \frac{1-2v}{1-v}$ and $\alpha_{j}\neq \qty(1-\frac{1+2\sqrt{1-8v(1-v)}}{2(1-v)})$ we have $\sqrt{1\pm q}=\sum\limits_{n=0}^{\infty}(-1)^{\frac{n}{2}\mp\frac{n}{2}} \ q^n  \ \binom{\frac{1}{2}}{n}$. Taking only up to second order, we write $\sqrt{1+q_1}\approx \qty(1+\frac{q_1}{2}-\frac{q_1^2}{8}+\mathscr{O}(q_1^3))$ and $\sqrt{1-q_2}\approx  \qty(1-\frac{q_2}{2}-\frac{q_2^2}{8}-\mathscr{O}(q_2^3))$. Now, putting such approximated values of $q_1$ and $q_2$ in Eq.~(\ref{xijc3}), we obtain the following:
\begin{eqnarray}\label{xijc31}
    \xi_j&\approx& v \ \qty(1+\frac{q_1}{2}-\frac{q_1^2}{8}+\mathscr{O}(q_1^3)) + (1-v) \ \qty(1-\frac{q_2}{2}-\frac{q_2^2}{8}-\mathscr{O}(q_2^3)) \nonumber \\
    &=& 1-\alpha_j^2 \frac{1}{8v(1-v)}+\alpha_j^3\qty(\frac{1}{8v(1-v)}-\frac{1}{2})-\alpha_j^4\qty(\frac{1}{8v(1-v)}-\frac{3}{8})+\mathscr{O}(\alpha_j^5)
\end{eqnarray}

Note that for $0<v<\frac{1}{2}$, $\xi_{j}$ is always upper bounded by one. Since, the values of $\alpha_j$ can be arbitrarily small (follows from Eq.$(3)$ of the main text), thus omitting $\alpha_{j}^{4}$ and higher order terms, the maximum value of $\xi_{j}$ is $\xi_j^{max}=\qty(1-\frac{\alpha^{2}_{j}}{2})$ for $v=\frac{1}{2}$, which is the first order approximation of $\sqrt{1-\alpha_{j}^{2}}$.

Therefore, when $0\leq \alpha_{j}\leq \frac{1-2v}{1-v}$ and $\alpha_{j}\neq \qty(1-\frac{1+2\sqrt{1-8v(1-v)}}{2(1-v)})$ for all $1\leq j\leq k$, the CHSH value between Alice and Bob$^k$ given by Eq.~(\ref{chshkc2um}) reduces to the following:
\begin{equation}\label{alicebobchshkum}
    \mathcal{J}^k =\frac{2 }{2^{k-1}}\Bigg[\cos\delta \ \prod\limits_{j=1}^{k-1} \Bigg\{2-\frac{\alpha_j^2 }{8v(1-v)}+\alpha_j^3 \ \qty(\frac{1}{8v(1-v)}-\frac{1}{2}) -\alpha_j^4 \ \qty(\frac{1}{8v(1-v)}-\frac{3}{8})+\mathscr{O}(\alpha_j^5)\Bigg\} +  \alpha_k \ \sin\delta \Bigg] 
\end{equation}
Now, if we choose the value of $\alpha_j$ as $0\leq \alpha_{j}\leq \frac{1-4 v+4v^2}{1-3v+3v^2}$, then it is straightforward to show that $\alpha_j^3\qty(\frac{1}{8v(1-v)}-\frac{1}{2})-\alpha_j^4\qty(\frac{1}{8v(1-v)}-\frac{3}{8})\geq0$ for $0< v< \frac{1}{2}$. This upper bound of $\alpha_{j}$ is justified since $\frac{1-2v}{1-v}\geq \frac{1-4v+4v^2}{1-3v+3v^2}>\qty(1-\frac{1+2\sqrt{1-8v(1-v)}}{2(1-v)})$ for $0< v< 1/2$. Thus, the approximation given by (\ref{xijc31}) is valid within the range $0\leq \alpha_{j}\leq \frac{1-4v+4v^2}{1-3v+3v^2}$ and $\alpha_{j}\neq \qty(1-\frac{1+2\sqrt{1-8v(1-v)}}{2(1-v)})$. Within this range, Eq.~(\ref{alicebobchshkum}) reduces to following:
\begin{equation} \label{lb2k}
    \mathcal{J}^k \geq 2 \ \qty[\cos\delta \ \prod\limits_{j=1}^{k-1} \qty(1-\frac{\alpha_j^2 }{16v(1-v)}) +  \frac{\alpha_k \ \sin\delta }{2^{k-1}}]
\end{equation}
with $ \mathcal{J}^1=2(\cos\delta+\alpha_1 \sin\delta)$ and $ \mathcal{J}^1>2$ implies $\alpha_1 > \tan\frac{\delta}{2}$.

Now, if the lower bound of the Bell value given by Eq.~(\ref{lb2k}) is greater than two then the Bell inequality is violated for arbitrary $k$. Which in turn, gives an lower bound on $\alpha_k$. If $\alpha_k$ satisfies the lower bound then Bell inequality is violated. In particular, we have the following
\begin{eqnarray}
\label{case3chshbobk}
  \alpha_{k} &>&\frac{2^{k-1}}{\sin\delta} \ \qty[1-\cos\delta \ \prod\limits_{j=1}^{k-1} \qty(1-\frac{\alpha_j^2 }{16v(1-v)})] \implies \mathcal{J}^k>2.
\end{eqnarray}
where $0\leq \alpha_{j}\leq \frac{1-4v+4v^2}{1-3v+3v^2}$, since we are using the approximation given by (\ref{xijc31}). It is straightforward to see that $0<\frac{\alpha_j^2 }{16v(1-v)}<1$ in the range $\qty(0.0580<v<\frac{1}{2})\cap \qty(0<\alpha_{j}\leq\frac{1-4v+4v^2}{1-3v+3v^2})$. Thus, similar to Theorem $1$ we we will show the following

\begin{thm}
For arbitrary $k$ and $0.0580<v<\frac{1}{2}$, there exists suitable values of $\delta \in(0,\frac{\pi}{2}]$ and a sequence $\{s_1,s_2,\ldots,s_k\}$, for which 
\begin{equation} \label{thm2mr}
    0\leq \frac{2^{l-1}}{\sin\delta} \ \qty[1-\cos\delta \ \prod\limits_{j=1}^{l-1} \qty(1-\frac{s_j^2 }{16v(1-v)})]<s_{l}\leq \frac{1-4v+4v^2}{1-3v+3v^2} \ \ \forall 2\leq l \leq k, \ s_1>\tan\frac{\delta}{2}
\end{equation}

\end{thm}

\begin{proof}
For this purpose, we consider the following form of $\{s_j:0\leq j\leq k\}$ satisfying (\ref{case3chshbobk}) for arbitrary $k$
\begin{equation}\label{genalphak}
s_{l} = \begin{cases}
    0 & \ \text{if} \ l=0\\
   (1+\epsilon)\tan\frac{\delta}{2}  & \ \text{if} \ l=1\\
    \frac{2^{l-1}(1+\epsilon)}{\sin\delta} \ \qty[1-\cos\delta \ \prod\limits_{j=0}^{l-1} \qty(1-\frac{s_j^2 }{16v(1-v)})] & \ \text{if} \ l>1 \ \text{and} \ s_{l-1}\leq 1
\end{cases}
\end{equation}
where $0 < \epsilon\leq 1$. To show there exists suitable values of $\delta$ such that $s_{l}$ given by Eq.~(\ref{genalphak}) is less than or equal to $\frac{1-4v+4v^2}{1-3v+3v^2}$ for all $1\leq l\leq k$, we will demonstrate as $\{s_{l}:1\leq l\leq k\}$ given by Eq.~(\ref{genalphak}) can be made arbitrarily small for all $1\leq l\leq k$. In particular, we then demonstrate that  $(i)$ the sequence given by Eq.~(\ref{genalphak}) is monotonically increasing i.e. $s_{1}<s_{2}<\ldots<s_{k}$ and $(ii)$ $s_{l}$ can be made arbitrarily small as $\delta$ approaches zero i.e. $\lim_{\delta\rightarrow 0^{+}}s_{j}= 0^{+} \ \forall 1\leq l\leq k$. Thus, for arbitrary $k$, we can choose infinitesimal small values of $m$ such that $0<s_{l}\leq \frac{1-4v+4v^2}{1-3v+3v^2}$ where $s_{l}$ is given by (\ref{genalphak}). Combining $(i)$ with $(ii)$ imply that $0<s_{1}<s_{2}<\ldots<s_{k}\leq \frac{1-4v+4v^2}{1-3v+3v^2}$ for arbitrary $k$. Thus, we can conclude, there exists feasible values of $s_{k}$ satisfying $\mathcal{J}^k>2$ for small $\delta$ and arbitrary $k$. In the following proofs of  $(i)$ and $(ii)$ are given.
\end{proof}

\subsubsection*{Proof of (i) -- Monotonicity of $s_l$ given by Eq.~(\ref{genalphak}):} \label{monotone}

\begin{proof}
Considering the following relation:
\begin{eqnarray}
\frac{s_{l}}{s_{l-1}} &=& \frac{2^{l-1}\qty[1-\cos\delta\ \prod\limits_{j=0}^{l-1} \qty(1-\frac{s_j^2 }{16v(1-v)})]}{2^{l-2}\qty[1-\cos\delta\ \prod\limits_{j=0}^{l-2} \qty(1-\frac{s_j^2 }{16v(1-v)})]} \ ; \ \ l\geq 2 \nonumber \\
&=& 2 \ \frac{1-\cos\delta\qty(1-\frac{s_{l-1}^2 }{16v(1-v)}) \prod\limits_{j=0}^{l-2} \qty(1-\frac{s_j^2 }{16v(1-v)})}{1-\cos\delta\ \prod\limits_{j=0}^{l-2} \qty(1-\frac{s_j^2 }{16v(1-v)})} \label{mono1}
\end{eqnarray}
 Therefore, the above Eq.~(\ref{mono1}) reduces to the following 
\begin{equation}
    \frac{s_{l}}{s_{l-1}} > 2 \ \frac{1-\cos\delta \prod\limits_{j=0}^{l-2} \qty(1-\frac{s_j^2 }{16v(1-v)})}{1-\cos\delta\ \prod\limits_{j=0}^{l-2} \qty(1-\frac{s_j^2 }{16v(1-v)})}
\end{equation}
 
Thus, in the range $v\in \qty(0.0580,\frac{1}{2})$ and $s_l\in \Big(0,\frac{1-4v+4v^2}{1-3v+3v^2}\Big]$, we have $s_{l} > 2s_{l-1}$, implying $s_{l}> s_{l-1} > s_{l-2} >\ldots>s_{1}>0$. 
\end{proof}

%%%%%%%%%%%%%%%%%%%%%%%%%%%%%%%%%%%%%%%%%%%%%%%%%%%%%%%%%%%%%%%%%%%%%%%%%%%%%%%%%%%%%%%%%%%%%%%%%%%%%%%%%%%%%%%%%%%%%%%%%%%%%%%%%%%%%%%%%%%%%%%%%%%%%%%%%%%%%%%%%%%%%%%%%%%%%%%%%%%%%%%%%%%%%%%%%%%%%%%%%%%%%%%%%%%%%%%%%%%%%%%%%%%%%%%%%%

\subsubsection*{Proof of (ii) --  $\lim_{m\rightarrow 0^{+}}s_l=0$, where $s_{l}$ is given by Eq.~(\ref{genalphak}):}\label{continue}

\begin{proof}

We consider an upper bound on the sequence given in Eq.~(\ref{genalphak}) and demonstrate that upper bound approaches zero as $m$ approaches zero. This, coupled with positivity of $s_l$ guarantees that $\alpha_{l}$ approaches zero as $\delta$ approaches zero. Since  $\cos \delta> 1-\frac{\delta^{2}}{2}$ and $\sin\delta > \delta$ for $\delta \in (0,\frac{\pi}{4}]$, we have 
\begin{eqnarray}
s_{l} &<& \frac{2^{l-1}(1+\epsilon)}{\delta} \ \qty[1-\qty(1-\frac{\delta^{2}}{2}) \ \prod\limits_{j=1}^{l-1} \qty(1-\frac{s_j^2 }{16v(1-v)})]
\end{eqnarray}
Now, let us define a new sequence
\begin{equation}
\beta_{l}(\delta) =\begin{cases}
    \frac{2^{l-1}(1+\epsilon)}{\delta} \ \qty[1-(1-\frac{\delta^{2}}{2}) \ \prod\limits_{j=1}^{l-1} \qty(1-\frac{\beta_j^2 }{16v(1-v)})] & \ \text{if} \ \beta_{l-1}<1 \\
    \infty & \ \text{otherwise}
\end{cases}
\end{equation}
Here $\beta_{1}=(1+\epsilon)\frac{\delta}{2}$. Similar procedure to (\ref{monotone}) yields $\beta_{l}>\beta_{l-1}>\ldots>\beta_{1}$. Positivity of $s_{l}$ and $s_{l}<\beta_{l}$ ensures, $\lim_{\delta\rightarrow 0^{+}}\beta_{l}=0$ implies $\lim_{\delta\rightarrow 0^{+}}s_{l}=0$. Next, in the following, by using the inductive argument, we show that $\lim_{\delta\rightarrow 0^{+}}\beta_{l}=0$.

To begin with, it is evident that $\lim_{\delta\rightarrow 0^{+}}\beta_{1}=0$. Then, for $l=2$, we have
\begin{equation}
\beta_{2} = 2(1+\epsilon)\Bigg[\frac{\delta}{2}\qty(1+\frac{(1+\epsilon)^{2}}{32v(1-v)})+\frac{(1+\epsilon)^{2}\delta^{3} }{128v(1-v)}\Bigg] \in \mathscr{P}_1 (\delta)
\end{equation}
where $\mathscr{Poly}_{n}(\delta)$ is a set of polynomial functions of $\delta$ with the lowest power of $\delta$ being $n$. A polynomial with constant term is within the set $\mathscr{Poly}_{0}(\delta)$. Now, let us assume that $\beta_j \in \mathscr{P}_1 (\delta) \ \forall j\in \{1,2, ..., l-1\}$. Next, we show that $\beta_{l}\in\mathscr{Poly}_{1}(\delta)$.  

Note that, since $\frac{\beta_{j}^{2}}{16v(1-v)}\in \mathscr{Poly}_{2}(\delta)$, the product $\prod\limits_{j=1}^{l-1} \qty(1-\frac{\beta_j^2 }{16v(1-v)})$ is expressed as $\qty[1-\mathcal{P}_2(\delta)]$, where $\mathcal{P}_2(\delta) \in \mathscr{Poly}_{2}(\delta)$. Then $\beta_l$ is given by
\begin{eqnarray}
    \beta_{l} &=&  \frac{2^{l-1}(1+\epsilon)}{\delta} \ \Bigg[1-\qty(1-\frac{\delta^{2}}{2})\Bigg(1-\mathcal{P}_2(\delta)\Bigg)\Bigg] \nonumber \\
    &=& 2^{l-1}(1+\epsilon) \ \Bigg[\frac{\delta}{2}+\frac{\mathcal{P}_2(\delta)}{\delta}-\frac{\delta \ \mathcal{P}_2(\delta)}{2}\Bigg] \in \mathscr{P}_1(\delta) \label{final1}
\end{eqnarray}
Now, from above Eq.~(\ref{final1}), it follows that $\lim_{\delta\rightarrow 0^{+}}\beta_{l}=0$. This and $0\leq s_{l}<\beta_{l}$ guarantee $\lim_{m\rightarrow 0^{+}}s_{l}=0$. 
\end{proof}

Since, $s_{l}$ can be made arbitrarily small one can always choose $s_{l}\leq \frac{1-4v+4v^2}{1-3v+3v^2}$. Therefore, completing the proof of existence of $\alpha_{l}$ satisfying $0\leq \frac{2^{l-1}}{\sin\delta} \ \qty[1-\cos\delta\ \prod\limits_{j=1}^{l-1} \qty(1-\frac{s_j^2 }{16v(1-v)})]<s_{l}\leq \frac{1-4v+4v^2}{1-3v+3v^2}$ within the range $0.058<v\leq \half$. Since, the Bell value as well as the lower bound discussed here is symmetric around half, we can replace $v$ by $1-v$ and obtain the full range of the allowed values of $v$ giving infinite sharing. Thus, proving the claim of Theorem $2$ in the main text.

If we follow the same line of reasoning for the $v=0$ case, then it can be shown that $\beta_{l}\in \mathscr{P}_{0}(\delta)$ implying $\lim_{m\rightarrow 0^{+}}\beta_{l}\neq 0$. Thus, it is not possible to prove the existence of infinite sharing of Bell nonlocality in this case following similar arguments.

%%%%%%%%%%%%%%%%%%%%%%%%%%%%%%%%%%%%%%%%%%%%%%%%%%%%%%%%%%%%%%%%%%%%%%%%%%%%%%%%%%%%%%%%%%%%%%%%%%%%%%%%%%%%%%%%%%%%%%%%%%%%%%%%%%%%%%%%%%%%%%%%%%%%%%%%%%%%%%%%%%%%%%%%%%%%%%%%%%%%%%%%%%%%%%%%%%%%%%%%%%%%%%%%%%%%%%%%%%%%%%%%%%%%%%%%%%%%%%%%%%%%%%%%%%%%%%%%%%%%%%%%%%%%%%%%%%%%%%%%%%%%%%%%%%%

  \end{widetext}
%%%%%%%%%%%%%%%%%%%%%%%%%%%%%%%%%%%%%%%%%%%%%%%%%%%%%%%%%%%%%%%%%%%%%%%%%%%%%%%%%%%%%%%%%%%%%%%%%%%%%%%%%%%%%%%%%%%%%%%%%%%%%%%%%%%%%%%%%%%%%%%%%%%%%%%%%%%%%%%%%%%%%%%%%%%%%%%%%%%%%%%%%%%%%%%%%%%%%%%%%%%%%%%%%%%%%%%%%%%%%%%%%%%%%%%%%%%%%%%%%%%%%%%%%%%%%%%%%%%%%%%%%%%%%%%%%%%%%%%%%%%%%%%%%%%%%%%%%%%%%%%%%%%%%%%%%%%%%%%%%%%%%%%%%%%%%%%%%%%%%%%%%%%%%%%%%%%%%%%%%%%%%%%%%%%%%%%%%%%%%%%%%%%%%%%%%%%%%%%%%%%%%%%%%%%%%%%%%%%%%%%%%%%%%%%%%%%%%%%%%%%%%%%%%%%%%%%%%%%%%%%%%%%%
  
 %\bibliographystyle{JHEP}
 \bibliography{Arxiv_Sub_Sharing}

\end{document}